\documentclass[11pt]{article}
\usepackage[T1]{fontenc}

\usepackage{amsmath,amssymb,mathrsfs}
\usepackage{mathtools}
\usepackage{graphicx}
\usepackage{tikz}
\usepackage[numbers]{natbib}
\usepackage{float}
\usepackage{subfig}
\usepackage{etoolbox}
\usepackage{todonotes}
\usepackage{hyperref}
\hypersetup{colorlinks,
pdfstartview = FitH,
bookmarksopen = true,
bookmarksnumbered = true,
linkcolor = blue,
plainpages = false,
hypertexnames = false,
urlcolor = blue,
citecolor = blue}
\usepackage{amsthm}
\usepackage{orcidlink}
\usepackage{thmtools} 
\usepackage{authblk}
\usepackage{thm-restate}
\usepackage{cleveref}
\newtheorem{theorem}{Theorem}
\newtheorem{lemma}{Lemma}
\newtheorem{proposition}{Proposition}

\newtheorem{definition}{Definition}

\DeclareMathOperator*{\argmin}{arg\!\min}
\DeclareMathOperator*{\PoA}{PoA}
\DeclareMathOperator*{\PoS}{PoS}

\newcommand{\N}{\mathbb{N}}

\newcommand{\badNash}{\ensuremath{\hat{S}}}
\newcommand{\Nash}{\ensuremath{\mathcal{S}_{\operatorname{eq}}}}

\newcommand{\daniel}[2][]
{\todo[color=yellow, #1]{#2}}

\newif\ifappendix
\appendixtrue
\newcommand{\gotoAppendix}[2]{
\ifappendix
	#1
\else
	#2
\fi
}

\begin{document}

\title{On the Price of Anarchy in Packet Routing Games with FIFO}

\date{}
\author[1]{Daniel Schmand\orcidlink{0000-0001-7776-3426}}
\author[1]{Torben Schürenberg\orcidlink{0009-0006-5947-0172}}
\author[2]{Martin Strehler\orcidlink{0000-0003-4241-6584}}
\affil[1]{University of Bremen, Germany, \textit {\{schmand,torsch\}@uni-bremen.de}}
\affil[2]{Westsächsische Hochschule Zwickau, Germany, \textit{martin.strehler@fh-zwickau.de}}
\maketitle              
\begin{abstract}
\noindent We investigate packet routing games in which network users selfishly route themselves through a network over discrete time, aiming to reach the destination as quickly as possible. Conflicts due to limited capacities are resolved by the first-in, first-out (FIFO) principle.
Building upon the line of research on packet routing games initiated by Werth et al.~\cite{WERTH201418}, we derive the first non-trivial bounds for packet routing games with FIFO. 
Specifically, we show that the price of anarchy is at most 2 for the important and well-motivated class of uniformly fastest route equilibria introduced by Scarsini et al.~\cite{DBLP:journals/ior/ScarsiniST18} on any linear multigraph. We complement our results with a series of instances on linear multigraphs, where the price of stability converges to at least $\frac{e}{e-1}$. 
Furthermore, our instances provide a lower bound for the price of anarchy of continuous Nash flows over time on linear multigraphs which establishes the first lower bound of $\frac{e}{e-1}$ on a graph class where the monotonicity conjecture is proven by Correa et al.~\cite{DBLP:journals/mor/CorreaCO22}.

\end{abstract}
\clearpage
\section{Introduction}
\subsection{Motivation}
Selfish routing games, also known as competitive routing games, have been the subject of extensive study due to their wide-ranging applications in traffic and communication networks. These games involve multiple agents, such as road users or data packets, competing for limited resources like network paths and bandwidth. When considering road traffic, the idea of modeling with a continuous flow model appears unrealistic. Given that vehicles are indivisible and possess substantial size, the applicability of particles that can be fractionated at will seems dubious with respect to capturing all relevant dynamics. Consequently, there has been a recent surge in the study of atomic, dynamic flow models.

Werth et al.~\cite{WERTH201418} considered such a model with deterministic queues, which also ensures the first-in, first-out (FIFO) property, from a game-theoretic perspective. The authors examine two variants, each pertaining to distinct objective functions. For narrowest flows, the cost of a user is the most expensive resource on her path and users try to minimize these bottlenecks. However, it is NP-hard to determine according flows and the price of anarchy (PoA) is tightly bounded by the number of players. Conversely, when dealing with the \emph{makespan} objective function, where every user strives to minimize their arrival time, equilibria are easy to compute in the single commodity case. However, the PoA has been left open. Since the introduction of this model by Werth et al., numerous variants have been studied, but the PoA in the original model is still an unresolved problem.

 This paper aims to take a step towards addressing this gap. We find a constant upper bound of 2 for the class of linear multigraphs regarding the makespan objective and provide a sequence of linear multigraphs that converge to a price of stability ($\PoS$) and therefore also a PoA of at least $\frac{e}{e-1}$.

\subsection{Informal Model Description}
Prior to delving into a comprehensive description of the model in Section~\ref{sec:prelim}, we offer an informal overview that is intended to provide a more effective categorization within the multifaceted realm of routing games.

We are examining a dynamic, competitive routing game involving atomic players on a directed graph $G=(V,E)$. Each edge $e \in E$ is equipped with an integral transit time $\tau(e) \in \N_{>0}$. Each player routes an unsplittable packet of unit size through this network from origin $s$ to destination $d$, that is, the players choose an $s-d$ path, start at time zero and try to reach $d$ as quickly as possible, i.e., players aim to minimize the arrival time at $d$. Furthermore, each edge is equipped with a point queue. At any point in time arbitrarily many packets may enter such a queue, but the queue has a limited outflow rate, that is at most $\nu_e$ packets may leave the queue at a particular time step. These queues operate according to the FIFO principle, ensuring that packets exit in the same sequence they entered. There is no limit to the maximum storage capacity of these queues. Consequently, the total latency on an edge is composed of the transit time and any potential waiting time in the queue.

In particular, we consider so-called uniformly fastest route (UFR) equilibria that have also been analyzed by~\cite{DBLP:journals/ior/ScarsiniST18}. Such a UFR equilibrium is achieved if for every player and for every node $v$ on its route, there exists no other route such that this player can arrive earlier at $v$. In other words, a UFR equilibrium also takes arrival time at intermediate nodes into account.
For the social optimum, we consider the makespan of the problem, that is, assuming all players start at time zero, we minimize the arrival time of the last player.

\subsection{Related Literature}

The diversity of routing games mirrors the breadth of their applications. We will give a brief overview of the origin and development, focusing particularly on the most relevant results for the present work.

A shared characteristic across all games is that each player possesses a utility function, typically characterized by parameters such as latency\footnote{In the literature latency is sometimes also called travel time.}~(see, e.g.,~\cite{braess1968paradoxon,kochskutella2011,pigou,DBLP:journals/ior/ScarsiniST18}, residence time~\cite{cao2022bounding}, bottleneck~\cite{WERTH201418}, or arrival penalties\footnote{ Deviation from the desired arrival time is often taken into account as an additive penalty in the objective of dynamic traffic assignment models (DTA) used in the traffic community.}~\cite{HAN201317} and chooses a path (strategy) from its origin to its destination. A state in the game is referred to as a (pure) Nash equilibrium when no individual player can decrease the private utility function by unilaterally altering their route.

Equilibrium solutions are often compared to a system optimum with respect to some objective function. 
Common objectives include minimizing the total latency (sum over all road users)~\cite{harks2018competitive}, the total completion time (also known as makespan)~\cite{WERTH201418}, the (average) delay~\cite{DBLP:journals/ior/ScarsiniST18}, 
or the throughput~\cite{kochskutella2011}. The ratio between the worst equilibrium and the system optimum is referred to as the price of anarchy ($\PoA$). In addition, we also consider the price of stability ($\PoS$), which refers to the ratio between the best equilibrium and the system's optimal state.  

In the categorization of such competitive routing games, it is necessary to make several differentiations. Initially, the games can be classified as either static or dynamic. In a static setup, players choose their path, and interference arises when both paths share the same edge. This interference can be attributed to load-dependent latencies~\cite{roughgarden2005selfish}, or due to the edge capacities that constrain the total flow on an edge~\cite{correa2004selfish}. Pigou's two link network~\cite{pigou} and Braess’ Paradox~\cite{braess1968paradoxon} are pioneering examples in the realm of static games.

Conversely, dynamic games incorporate the element of time, with particles navigating the network in a time-sensitive manner. Typically, the subsequent player may experience delays due to the preceding player, as the former must wait until the latter clears the next edge. Furthermore, it is common to require adherence to the FIFO principle. Moreover, the capacity typically constrains the flow rate (measured in flow units per time unit), rather than the aggregate flow passing through an edge. A notable early and often adapted example is Vickrey’s queuing model~\cite{vickrey1969congestion}, which introduced point queues of no physical space and factored in varying departure times. 
Since then, the game-theoretic aspects of flows over time have been a topic of extensive discussion within the traffic community. From the perspective of algorithmic game theory, Koch and Skutella~\cite{kochskutella2011} have demonstrated that Nash flows over time with the underlying deterministic queuing model can be interpreted as a sequence of specific static flows. 

The second differentiation pertains to non-atomic versus atomic games. In non-atomic flows, one can envision an infinite number of players, each controlling an infinitesimally small packet. 
On the other hand, atomic flows, as introduced by Rosenthal~\cite{rosenthal1973network}, are characterized by each entity controlling a substantial portion of the total flow. Typically, these games involve a finite number of players, each with an unsplittable packet of unit size.

Non-atomic models are relatively well-understood due to their amenability to analytical methods. 
For example, in non-atomic models, the value of an equilibrium flow is often unique~\cite{cominetti2011existence}.
On the other hand, the $\PoA$ is still an open question in many dynamic scenarios. A significant advancement was made by Correa et al.~\cite{DBLP:journals/mor/CorreaCO22}, where they established that if one is permitted to restrict the inflow rate of the equilibrium to the optimum flow’s initial inflow rate, the $\PoA$ concerning the makespan is bounded by $\frac{e}{e-1}$. Unfortunately, these findings are contingent upon a monotonicity conjecture~\cite{DBLP:journals/mor/CorreaCO22}. Despite its intuitive appeal, this conjecture has only been validated for relatively straightforward graph classes, such as linear multigraphs, to date.

The suitability of non-atomic models, particularly for depicting scenarios such as road traffic, is a subject of ongoing debate, since it is unclear whether particle size is irrelevant in such applications.
A game-theoretic approach based on a dynamic, discrete model with point queues was introduced by Werth~\cite{WERTH201418}.
Subsequent to this approach, numerous variants of atomic routing games have been explored, given that these models necessitate the incorporation of additional concepts or intricate modeling details. A particular aspect is simultaneity: what happens when two players want to use the same edge at the same time? This necessitates a concept for parallel processing (fair time sharing, see~\cite{hoefer2011competitive}) or tie-breaking rules. For the latter, predefined priorities on players~\cite{harks2018competitive,WERTH201418} or edges~\cite{DBLP:conf/sigecom/CaoCCW17,scheffler2022routing,WERTH201418} are typically employed. 

Atomic models are often more challenging to analyze since they usually do not have unique equilibria, but some results were obtained for special setting. For example, Scarsini et al.~\cite{DBLP:journals/ior/ScarsiniST18} studied a game where the inflow occurs in generations of players, also emphasizing the relevance of UFR equilibria. Similarly, Cao et al.~\cite{DBLP:conf/sigecom/CaoCCW17} studied a game where the inflow rate never exceeds the network capacity, adding the concept of dynamic route choices.

However, the relationship between discrete and continuous routing models remains partially elusive. When considering system optimal solutions for a single source and a single sink, the similarities are pronounced. Already Ford and Fulkerson~\cite{ford1958constructing} pioneered the computation of dynamic flows for discrete time based on static flow computations. Fleischer and Tardos~\cite{DBLP:journals/orl/FleischerT98} successfully addressed this problem for continuous time. A comprehensive summary is available by Skutella~\cite{DBLP:conf/bonnco/Skutella08}. The situation becomes more intricate when examining equilibria. A Nash equilibrium in the continuous case is an equilibrium in the discrete case, if and only if the continuous flow can be interpreted as a flow of discrete packets, which is not possible in general. Conversely, while every state in the discrete setting can be interpreted as a feasible continuous flow~\cite{DBLP:journals/orl/FleischerT98}, the equilibrium property is often lost. That is, a Nash equilibrium in the discrete case does not necessarily translate to an equilibrium in the continuous case.

\subsection{Our Contribution}

We consider an atomic dynamic routing game on linear multigraphs\footnote{This graph class is also known in the literature as chain-of-parallel networks~\cite{DBLP:conf/sigecom/CaoCCW17,DBLP:journals/ior/ScarsiniST18}.}, where we have a linearly ordered node set $V=\{v_1,\dots,v_m\}$ and the edge set consists of multiple edges $(v_j,v_{j+1})$. For this graph class, we prove the $\PoA$ to be at most 2, presenting the first result of a constant $\PoA$ on a graph class completely independent of the number of players or the network size. Specifically, we do not impose artificial restrictions on the inflow rate of the network. Furthermore, we show that whenever players have the choice between multiple quickest paths, the sum of queues in the network is maximized when players greedily opt for the longest queues available.

Moreover, we show that there are instances of linear multigraphs, where the $\PoS$ converges to at least $\frac{e}{e-1}$, indicating that discrete Nash flows even on such a restricted graph class are non-trivial. Furthermore, this result aligns with the upper bound for the $\PoA$ in a non-atomic routing game investigated by Correa et al.~\cite{DBLP:journals/mor/CorreaCO22} where the inflow rate is restricted to that of an optimal flow and a certain monotonicity conjecture holds. 
We show that our instances can be translated to the setting of Correa et al., thereby establishing a lower bound of $e/(e-1)$ on the $\PoA$ in their context. With the monotonicity conjecture validated for these instances by the authors of~\cite{DBLP:journals/mor/CorreaCO22}, we achieve the first tight bound on the $\PoA$ for a non-trivial class of graphs concerning Nash flows over time.

This paper is organized as follows. Upon establishing the required notation, we proceed to examine system optimal flows within our framework in Section~\ref{sec:understandingopt}. Subsequently, we demonstrate the unique structure of the least favorable UFR equilibrium and establish an upper bound of 2 on the $\PoA$ on linear multigraphs in Section~\ref{sec:poa}. We wrap up our findings in Section~\ref{sec:pos} by formally showing a strong connection to the dynamic model and by establishing a lower bound on the $\PoS$ of $\frac{e}{e-1}$ for this class of graphs. Please note that some proofs have been omitted and are provided in the appendix.
\section{Preliminaries}\label{sec:prelim}
We consider the classic packet routing game introduced by~\cite{WERTH201418}. Formally, a packet routing game $\Gamma$ is given by the tuple $\Gamma =(N, G=(V,E), \tau: E \rightarrow \N_{>0},$ $s \in V, d \in V)$,
where $N=[n]\coloneq \{1, \dots, n\}$ with $n\in \N_{>0}$ is the player set, $G$ is a multigraph with nodes $V$ and directed edges $E$. Additionally, $\tau(e) \in \N_{>0}$ denotes the transit time of an edge $e \in E$. Each player has one packet located at the source node $s \in V$ and aims to send this packet as fast as possible to the destination node $d$.

In this work, we focus on packet routing games on \emph{linear multigraphs} and show the first non-trivial bounds on the $\PoA$ for packet routing games with FIFO policy. A linear multigraph for $s\neq d$ is given by $G=(V,E)$ with \mbox{$V=\lbrace s=v_0, v_1, \ldots , v_m = d\rbrace$}  and $j=j'+1$ for all edges $(v_{j'},v_j)\in E$. We call all edges from $v_{j-1}$ to $v_{j}$ together with the nodes $v_{j-1}$ and $v_{j}$ the $j$-th layer of the graph. 
We assume the edges $e^j_1, e^j_2, \ldots$ in each layer $j$ to be ordered with respect to their transit times $\tau(e^j_1)\leq \tau(e^j_2)\leq \ldots$. 
We denote the set of all linear multigraphs by $\mathcal{G}$.
 
Since each player is associated with a single packet, we will refer to the packet of a player as the player itself. Given $G$, $s$, and $d$, the strategy space $\mathcal{S}\coloneqq \mathcal{P}^n$ for the players is given by all simple $s-d$ paths $\mathcal{P}$ in $G$. A tuple $S=(P_{(1)}, \dots, P_{(n)}) \in \mathcal{S}$ is called a state of the game. Each state $S \in \mathcal{S}$ induces a network loading in the following sense.

\paragraph{Network Loading:}
Given a state $S$ of the game, we define the positions of every player at any point in time by the following algorithm: 
We initialize at time $t=0$ all queues $(q_e)_{e\in E}$ to be empty and position all packets at $s$. We add each player to the queue of the first edge in their paths. Players in queues are ordered according to FIFO, i.e., by their arrival time at $e$, where we break ties according to the player's index in favor of the player with the lower index.
We now iterate over $t$ until all packets have arrived at $d$:
For each edge $e$ we add all players $i$ with $e \in P_{(i)}$ who have left the queue of their previous edge $e'$ on their path at time $t-\tau(e')$ to queue $q_e$. Subsequently, for every edge $e$ with a non-empty queue, we remove the first player from the queue.\footnote{Note that we assume edge capacities to be equal to 1 to simplify notation, but we can extend the model to arbitrary edge capacities as we discuss as we discuss in Appendix~\ref{sec:edgcap}.} Hence, the player arrives at $\tau(e) + t$ at the head of $e$. If a player $i$ is both added to $q_e$ and removed from $q_e$ at the same time $t$, we say \emph{player $i$ does not queue on $e$}. Furthermore, when a player leaves the queue of an edge at time $t$, we refer to this edge as \emph{used (at time $t$)}. Note that this procedure can be turned into a polynomial time algorithm by keeping track and iterating only over relevant times $t$, where players leave queues.

We now introduce some additional notation for the network loading:
The \emph{waiting time} $w_e^i(S)$ of player $i$ on edge $e$ is defined as the time that player $i$ spends in queue $q_e$. The \emph{latency} $l_e^i(S)$ of player $i\in N$ on edge $e \in E$ is defined by $l_e^i(S) \coloneqq \tau(e) + w^i_e(S)$. 
The \emph{workload} $l_{e}(S,t)$ on an edge $e$ at time $t$ under the strategy profile $S$ is defined as the latency that a fictional player entering $e$ at time $t$ would experience on that edge if she had the largest index among all players present at $q_e$ at that time.
For the chosen $s-d$ path $P_{(i)} = (e_1,\ldots, e_m)$ of player $i$, the arrival time $a_{v_j}^i(S)$ of $i$ at the head $v_j$ of the edge $e_j$ is given by $a_{v_j}^i(S) = \sum_{r=1}^{j} l_{e_r}^i(S)$. For a fixed strategy profile $S$, this yields a uniquely defined \emph{arrival pattern} $a_v(S) = (a_v^i(S))_{i \in N}$ for every node $v\in V$, which we interpret as an ordered vector of the arrival times of all players for that node. For $v=s$, we set $a_s(S)= 0^n\coloneq(0,\ldots,0)\in \N_{\geq 0}^n$, since all players start at time $0$ at $s$. We call the number of players that arrive at node $v$ at time $t$ the \emph{inflow of $v$ at time $t$}. We refer to all players with arrival time $t$ at node $v$ as the \emph{generation $t$ at $v$}. We define the \emph{completion time} of player $i\in N$ to be $C_i(S)\coloneq a_{d}^i(S)$ and the (total) \emph{completion time} of the state to be $C(S)\coloneq \max_{i\in N} C_i(S)$. Since the network loading algorithm uniquely determines the positions of all players at any given time, we can compute all these values through a post-processing step.

\begin{figure}[t]
    \centering
        \tikzstyle{vertex}=[circle,fill=black!25,minimum size=20pt,inner sep=0pt]
        \tikzstyle{smallvertex}=[circle,fill=black,minimum size=6pt,inner sep=0pt]
        \tikzstyle{edge} = [draw,thin,->]
        \tikzstyle{weight} = [font=\small]
        \tikzstyle{selected edge} = [draw,line width=3pt,-,red!50]
        \tikzstyle{big edge} = [draw,line width=2pt,->,black]

        \begin{tikzpicture}[scale=0.55,auto,swap]
            \node[vertex](s1) at (0,0){$s$};
            \node[vertex](s2) at (3,0){$v$};
            \node[vertex](s3) at (7,0){$d$};

            \draw [edge] (s1) to[out=45,in=135, distance=1cm ] node[weight,above,black]{$\tau(e_1^1)=1$} (s2);
            \draw [edge] (s1) to[out=-45,in=-135, distance=1cm ] node[weight,below=0.2,black]{$\tau(e_2^1)=2$} (s2);

            \draw [edge] (s2) to[out=0,in=180, distance=1cm ] node[weight,below=0.25,black]{$\tau(e_1^2)=2$} (s3);

            \node[draw, fill=white] (P1) at (3.75,0) {$1$};
            \node[draw, fill=white] (P3) at (0.6,0.6) {$3$};
            \node[draw, fill=white] (P2) at (1.5,-1) {$2$};
        \end{tikzpicture}
\hspace{1cm}
        \begin{tikzpicture}[scale=0.55,auto,swap]
            \node[vertex](s1) at (0,0){$s$};
            \node[vertex](s2) at (3,0){$v$};
            \node[vertex](s3) at (7,0){$d$};

            \draw [edge] (s1) to[out=45,in=135, distance=1cm ] node[weight,above,black]{$\tau(e_1^1)=1$} (s2);
            \draw [edge] (s1) to[out=-45,in=-135, distance=1cm ] node[weight,below=0.2,black]{$\tau(e_2^1)=2$} (s2);

            \draw [edge] (s2) to[out=0,in=180, distance=1cm ] node[weight,below=0.25cm,black]{$\tau(e_1^2)=2$} (s3);

            \node[draw, fill=white] (P1) at (5,0) {$1$};
            \node[draw, fill=white] (P2) at (3.75,0) {$2$};
            \node[draw, fill=white] (P3) at (3.75,0.75) {$3$};
        \end{tikzpicture}
    \caption{An example to illustrate the notation. The left figure depicts the positions of players at $t=1$, the right figure for $t=2$.}
    \label{fig:SimpleExample}
\end{figure}
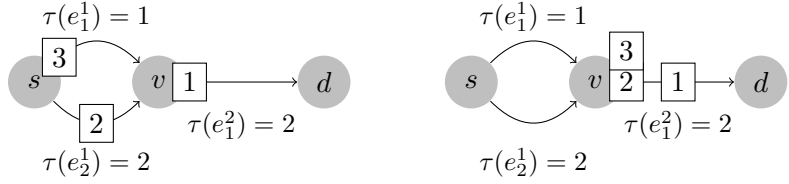

\paragraph{Social Objective:} 
The social objective is to minimize the completion time $C(S)$ among all $S\in \mathcal{S}$.
A state $S^*\in\mathcal{S}$ is called an \emph{optimal state} for the game if $S^* \in \argmin_{S \in \mathcal{S}} C(S)$.

\paragraph{Example:} 
To illustrate notation, we refer the reader to Figure~\ref{fig:SimpleExample}. For the player set $N=[3]$ and the strategies $S=((e_1^1, e_1^2),(e_2^1, e_1^2),(e_1^1, e_1^2))$ we obtain the waiting times $w_{e_1^1}^1(S)=w_{e_2^1}^2(S)=0$, $w_{e_1^1}^3(S)=1$ in the first layer. The latencies in the first layer are $l_{e_1^1}^1(S)= 1, l_{e_2^1}^2(S)=l_{e_1^1}^3(S)=2$ concluding in the workloads of $l_{e_1^1}(S,0)= l_{e_2^1}(S,0)=3$ and the arrival pattern $a_v(S)=(1,2,2)$ at $v$.
 For the second layer we obtain waiting times $w_{e_1^2}^1(S)=w_{e_1^2}^2(S)=0$, $w_{e_1^2}^3(S)=1$, latencies $l_{e_1^2}^1(S)= l_{e_1^2}^2(S)=2$, $l_{e_1^2}^3(S)=3$, workloads $l_{e_1^2}(S,0)=2$, $l_{e_1^2}(S,1)=3$, $l_{e_1^2}(S,2)=4$, $l_{e_1^2}(S,3)=3$ and the arrival pattern at $d$ is $a_d(S)=(3,4,5)$ with completion time $C(S)=5$.

\paragraph{Equilibria:}
A state $S=(P_{(1)}, \dots, P_{(n)}) \in \mathcal{S}$ is called a uniformly fastest route \emph{(UFR) equilibrium} if for every player $i\in N$, $P\in\mathcal{P}$ and node $v\in P_{(i)}\cap P$ it holds that $a_{v}^i(S)\leq a_{v}^i(S')$ for $S'=(P_{(1)}, \dots, P_{(i-1)}, P, P_{(i+1)}, \dots, P_{(n)})$.
 This means that in a UFR equilibrium, no player can achieve a better arrival time at any node $v$ on her chosen path by unilaterally altering her strategy to another path. Following the work of Scarsini et al.~\cite{DBLP:journals/ior/ScarsiniST18}, we believe this definition aptly
reflects the self-interested behavior of players in road traffic and we consider only this kind of equilibria, here\footnote{Details can be found in Appendix \ref{sec:equi}.}. The set of all such UFR equilibria is denoted by $\Nash$. To improve readability, we often drop the term \emph{UFR} and simply call them \emph{equilibria} in the rest of the paper. It should be noted that this set is not empty. Assume all players select their paths sequentially in the order of their indices, starting with the player with the smallest index, and every player chooses a uniform fastest route given the choices of the players with lower indices. Since none of the later players can arrive at an intermediate node before a player with lower index, no player can displace an earlier player and we indeed obtain an equilibrium.

We start by observing that the completion time in any equilibrium of our packet routing game is realized by the player with the highest index. 
\begin{lemma}\label{lemma_order}
    For every packet routing game $\Gamma$ on a linear multigraph and $S\in \Nash$ it holds that $C(S)= C_n(S)$. 
\end{lemma}
\begin{proof}
We will show the even stronger statement that $a_{v}^i(S)\leq a_{v}^{i+1}(S)$ at every node $v\in V$. Assume this would not hold. Let $v_j$ be the node with smallest index with $a_{v_j}^i(S)> a_{v_j}^{i+1}(S)$. Assume that player $i$ would change her strategy only in this layer by choosing the same edge $e$ as player $i+1$. Since $a_{v_{j-1}}^i(S)\leq a_{v_{j-1}}^{i+1}(S)$ the player $i$ enters the queue $q_e$ not later than player $i+1$. Since the state is the same up to node $v_{j-1}$ the time that other players that use edge $e$ arrive at edge $e$ does not change. Hence, player $i$ is removed from the queue not later than player $i+1$ is removed in state $S$ from this queue and thus arrives strictly earlier at $v_j$ than in $S$. This contradicts $S\in \Nash$.
\end{proof}
\noindent In essence, the proof of Lemma~\ref{lemma_order} also shows that all equilibria share a common characteristic with the specially constructed equilibrium above: all players arrive at every node and in particular at the destination $d$ in the order of their indices.

\paragraph{Price of Anarchy and Price of Stability:}
In terms of the social objective, our goal within the game $\Gamma$ is to minimize the completion time $C(S)$ across all $S\in \mathcal{S}$, i.e., the arrival time of the player arriving last at the destination node $d$. Since the players only optimize their own arrival times and not the social objective function, this equilibrium might not be efficient with respect to this social objective function.
Two popular measures of inefficiency are the price of anarchy ($\PoA$) and the price of stability ($\PoS$) which for a given packet routing game $\Gamma$ and an optimal state $S^*$ are defined by
\begin{align*}
    \PoA(\Gamma) \coloneqq \frac{\max\limits_{S \in \Nash} C(S)}{\min\limits_{S \in \mathcal{S}} C(S)} =\frac{\max\limits_{S \in \Nash} C(S)}{C(S^*)}, \; \PoS(\Gamma) \coloneqq \frac{\min\limits_{S \in \Nash} C(S)}{\min\limits_{S \in \mathcal{S}} C(S)} =\frac{\min\limits_{S \in \Nash} C(S)}{C(S^*)}
\end{align*}
and for a set $\mathcal{H}$ of graphs as
\begin{align*}
    \PoA(\mathcal{H}) \coloneqq \sup\limits_{\Gamma: G(\Gamma) \in \mathcal{H}} \PoA(\Gamma), \quad \PoS(\mathcal{H}) \coloneqq \sup\limits_{\Gamma: G(\Gamma) \in \mathcal{H}} \PoS(\Gamma).
\end{align*}
By definition we have $\PoS(\mathcal{H})\leq \PoA(\mathcal{H})$.
If the underlying game $\Gamma$ is not immediately apparent from context, it will be indicated through a superscript to all definitions.
\paragraph{Understanding an Optimal State $S^*$:}
\label{sec:understandingopt}
The first step in bounding the $\PoA$ and the $\PoS$ is understanding an optimal state. The optimal state for a packet routing problem sends all packets as fast as possible to the destination $d$. A very related problem is the discrete maximum dynamic flow problem, where as many packets as possible are being sent to $d$ in a given time horizon. Ford and Fulkerson introduced this problem in~\cite{ford1958constructing} and showed how to solve it by a single static flow computation. This static flow is then repeated over time to obtain a flow maximizing the number of packets arriving at $d$ within a given time horizon $T$. We will revisit these classical results to provide the following structure of an optimal state. 

\begin{restatable}{proposition}{satzOptStructure}\label{thm:opt}
For a packet routing game $\Gamma$ on $G\in\mathcal{G}$ let $P^1$ be a shortest $s-d$ path in $G$, and let $P^j$, $j>1$ be a shortest $s-d$ path in $G$ after the deletion of $P^1$, \dots, $P^{j-1}$. Let $k$ be such that $\tau(P^{k}) \coloneqq \sum_{e\in P^k} \tau(e)\leq C(S^*)$ and $\tau(P^{k+1})> C(S^*)$ if $P^{k+1}$ exists. There is an optimal state $S^*\in\mathcal{S}$ with the following properties.
    \begin{enumerate}
        \item There are no queues in any layer $\ell\geq 2$.
        \item $C(S^*)+1-\tau(P^1)$ packets use $P^1$ and $C(S^*)+\delta_j -\tau(P^j)$ packets use $P^j$ with $\delta_j\in\{0,1\}$ for $2\leq j \leq k$ and $\sum_{j=2}^k \delta_j = n-k\cdot C(S^*) -1 +\sum_{j=1}^k \tau(P^j)$.
    \end{enumerate}
\end{restatable}

\section{Upper Bound on the Price of Anarchy}\label{sec:poa}
In this section, we will show that the $\PoA$ for the class of linear multigraphs is bounded above by 2.
We first focus on subgraphs consisting of only two consecutive nodes. This approach allows us to derive certain structural properties that will be instrumental in our proof. Furthermore, we extend the model by allowing a distinct at $s$ that differs from the zero vector, i.e., not all players may be present at $s$ at the beginning. Instead, they arrive according to the schedule dictated by the starting pattern. Formally, players will enter the queue of the first edge of their path at their respective starting times.

In linear multigraphs, the sequence of nodes from $s$ to $d$ is uniquely defined. The strategy of a player in an equilibrium always consists of edges with the shortest latency to the next nodes in the sequence, given the strategy of all players of lower index. The only difference between different equilibria is the choice of the edge to use when multiple edges have the same latency for that player. A policy $\phi(\Gamma,a_s)$ is a mapping that maps a game $\Gamma$ and a starting pattern $a_s$ of the players set at node $s$ to a state of the game that is an equilibrium. 

Our attention will be centered on a specific policy, referred to as the \emph{greedy queue policy} $\hat{\phi}(\Gamma,a_s)$. Under this policy, the players sequentially choose their edges and the current player always opts for an edge with smallest latency. In case of a tie, each player decides for one of these edges with the longest queue. If there are still multiple edges with these properties, we break ties in favor of the edge with the lowest index among these options. We denote the unique strategy vector arising from the greedy queue policy by $\badNash=\hat{\phi}(\Gamma,a_s)$ for any given packet routing game $\Gamma$ and starting pattern $a_s$. 

In the subsequent discussion, we aim to demonstrate that $\badNash$ simultaneously maximizes the arrival times of all players at $d$ and we will bound the difference to an optimal state. To achieve this goal, we start by showing that the greedy queue policy $\badNash=\hat{\phi}(\Gamma,0^n)$ produces the worst completion time in any linear multigraph setting, i.e., $C(\badNash) = \sup_{S \in \Nash} C(S)$ for every $\Gamma$ (Proposition \ref{prop_together}). Afterwards, we will observe that removing edges from a network only enlarges the completion time, when players choose $s-d$ paths according to the greedy queue policy (Lemma \ref{lemma_adding edges}).
This will allow us to bound the $\PoA$ by first fixing an optimal state and, second, deleting all edges from $G$ that are not used in the optimum. The equilibrium in the resulting graph will utilize, at most, the edges that are employed in the optimal state, thereby enabling us to establish an upper bound for the arrival times (Theorem~\ref{satz_PoA2}).

We start by providing three fundamental properties of $\badNash$.
The first one states that the workload on parallel used edges differs by at most one. This can be used to show the second property that the greedy queue policy maximizes the individual arrival times in a single-layer network. The third lemma guarantees that a pointwise smaller starting pattern will never worsen the completion time of any player with the greedy queue policy.

\begin{restatable}{lemma}{lemmaAli}\label{lemma_ali}
Let $\Gamma$ be a packet routing game where $G(\Gamma)=(V=\{s,d\},E)$ represents a linear multigraph with a single layer and let $a_s\in \N_{\geq 0}^n$ be a starting pattern at $s$. Furthermore, we require the players to adhere to the strategy profile \mbox{$\badNash=\phi(\Gamma,a_s)$}. 
If two players leave the queues on edges $e^1_i$ and $e^1_j$, $i<j$, at the same time $t$, then for the workload it holds that $l_{e^1_j}(\badNash,t) \leq l_{e^1_i}(\badNash,t) \leq l_{e^1_j}(\badNash,t)+1$.
\end{restatable}



\begin{restatable}{lemma}{lemmaOneLayerNew}\label{lemma_onelayernew}
Let $\Gamma$ be a packet routing game with $G(\Gamma)=(V=\{s,d\},E)$ being a linear multigraph with only one layer and $a_s\in \N_{\geq 0}^n$ be a starting pattern at $s$. If players adhere to the strategy profile $\badNash=\phi(\Gamma, a_s)$, then $C_i(\badNash)\geq C_i(S)$ for all $S \in \Nash$ and all $i\in N$.
\end{restatable}

\begin{restatable}{lemma}{lemmaMnKStartpattern}\label{lemma_mnK_startpattern}
Let $\Gamma$ be a packet routing game with $G(\Gamma)=(V=\{s,d\},E)$ and $a, b\in \N_{\geq 0}^n$ be starting patterns at $s$ with $a\leq b$ pointwise. Let $\badNash^a = \phi(\Gamma, a)$ and $\badNash^b = \phi(\Gamma, b)$ denote the equilibria arising from the policy greedy queue for $a$ and $b$, respectively. For every $i \in N$ denote by $C_i(\badNash^a)$ and $C_i(\badNash^b)$ the corresponding completion times of player $i$. Then, $C_i(\badNash^a)\leq C_i(\badNash^b)$ for all $i\in N$.
\end{restatable}
\gotoAppendix{\noindent
Although the statement of Lemma \ref{lemma_mnK_startpattern} appears very natural, the proof is rather technical. 
The main idea of the proof is to restrict on patterns, which only differ in one position by one. Then, we introduce a dummy player and carefully swap this dummy player with consecutive players to transform one pattern into the other, while keeping track of the changes in the arrival patterns.}{
\begin{proof}
It is sufficient to show that the statement holds for starting patterns $a$ and $b$ which differ in only one entry by one. Then, a general $a$ can be transformed into $b$ by step-wise transformations. 

Let $a$ and $b$ be such starting patterns and $\ell$ be the index at which the patterns differ, that is $a(\ell)+1=b(\ell)$. Note that player $\ell$ is the player with the largest index of her generation in $a$ and the player with the smallest index of her generation in $b$. Obviously, $C_{\ell}(\badNash^a)\le C_{\ell}(\badNash^b)$, since all players up to player $\ell-1$ choose the same strategy under the greedy queue policy and $\ell$ could use the same edge under both starting patterns. Note that $C_\ell(\badNash^a) = C_\ell(\badNash^b)$ if and only if $\ell$ has to wait under starting pattern $a$, i.e., $\ell$ uses an old edge with a non-empty queue.

\paragraph{Claim:} If $C_\ell(\badNash^a) = C_\ell(\badNash^b)$, we have $\badNash^a = \badNash^b$.

\noindent \textit{Proof of claim.} Both $\badNash^a$ and $\badNash^b$ arise from the greedy queue policy $\phi$, i.e., is the same for all players up to $\ell-1$ as the graph and the starting times are the same up to player $\ell-1$. Given this setup, $\phi$ allocates the edge with the smallest workload to $\ell$ where ties are broken in favor of longer queues. As observed above, $\ell$ is allocated to some edge $\badNash^a_{(\ell)}$ with queue under $\badNash^a$. The arrival time of $\ell$ on the same edge under $b$ is obviously the same, thus, this is still the fastest edge with the longest queue under $b$. Hence, $\badNash^a_{(\ell)} = \badNash^b_{(\ell)}$ with the same arrival time of $\ell$, i.e., nothing changes for all subsequent players and $\badNash^a = \badNash^b$.\qed

\medskip
\noindent This finishes the proof for the case $C_\ell(\badNash^a) = C_\ell(\badNash^b)$.
On the other hand, if player $\ell$ arrives earlier because she starts one time unit earlier at $s$, then she does not queue, i.e., she uses an unused edge. As this edge is also unused when $\ell$ starts one time step later, we observe $C_\ell(\badNash^a) = C_\ell(\badNash^b) -1$. It remains to show that $C_i(\badNash^a) \leq C_i(\badNash^b)$ for all $i > \ell$. 

We add an additional dummy player to the player set with an index between $\ell$ and $\ell+1$, where the starting time of the dummy is the same as of $\ell+1$. As we consider the greedy queue policy and $\ell$ did not queue, the dummy uses the edge of player $\ell$ under $b$. Since the subsequent players now face the exact same situation in our new game with starting pattern $a$ and the additional dummy as in $b$, $C_i(\badNash^a) = C_i(\badNash^b)$ for all $i\geq\ell +1$.

Now, we will carefully modify the priority of the dummy and swap it with the subsequent players generation by generation. While doing so, we observe that the arrival times $C_i(\badNash^a)$ do not increase. For the generation of player $\ell +1$, we sequentially swap the priorities and strategies of the dummy with its subsequent player. This swap only changes the strategies of the dummy and the subsequent player and may decrease the arrival time of the subsequent player. As this swap does not change queue lengths, all other players still behave according to the greedy queue policy. Thus, the new state is an equilibrium according to greedy queue. We can continue this procedure until the dummy player is the player with the largest index of its generation. Note that the swaps might have increased the dummy player's arrival time. 

If the dummy has the largest index in its generation, we have two cases. 

\paragraph{Case 1:} If the dummy player is in a queue, we increment its generation by one, i.e., it is now the player with the smallest index of the next generation and starts one time unit later. Now, the same argument as in the claim holds for the dummy and she chooses the same edge under the policy greedy queue. We continue and iterate the swapping until the dummy is again the player with the largest index of the generation.
\paragraph{Case 2:} If the dummy player uses an edge without queue or is the player with the largest index, we delete the dummy player. In the first subcase, deleting the dummy does not change the queues for the player after the dummy, as this player could even take the dummy's edge without observing any queue. In the second subcase, the dummy was the player with the largest index, and deleting the dummy trivially does not change $C(\badNash^a)$.\\

\noindent Finally, we have started with $C_i(\badNash^a) = C_i(\badNash^b)$ for all $i>\ell$ and we have never increased $C_i(\badNash^a)$ during the dummy swaps. Thus, we have $C_i(\badNash^a)\leq C_i(\badNash^b)$ for all $i\in N$. 
\end{proof}}

Now we are ready to show that the greedy queue policy constructs an equilibrium simultaneously maximizing the arrival times of all players under all equilibria.
\begin{proposition}\label{prop_together}
Let $\Gamma$ be a packet routing game with $G(\Gamma)\in\mathcal{G}$ (and starting pattern $0^n$). The greedy queue policy constructs an equilibrium maximizing the arrival times of all players at $d$ under all equilibria $S \in \Nash$. Specifically, for $\badNash=\phi(\Gamma,0^n)$, it holds that $C(S) \leq C(\badNash)$.
\end{proposition}
\begin{proof}
Fix an arbitrary equilibrium $S \in \Nash$.
We will show this statement by carefully transforming $S$ to $\badNash$ in $m$ steps.
In the $j$-th iteration of the transformation, the players follow the strategy of $S$ until node $v_{m-j}$ and then follow the greedy queue policy. Denote this strategy as $S^j$.
Observe that $S=S^0$ and $\badNash=S^m$.
For any fixed $j\in \{0,\ldots,m-1\}$ the strategies $S^j$ and $S^{j+1}$ are identical until node $v_{m-j-1}$ and therefore also the arrival patterns at $v_{m-j-1}$. With Lemma $\ref{lemma_onelayernew}$, we know that the arrival pattern of $S^{j+1}$ at node $v_{m-j}$ is pointwise at least as big as the arrival pattern of $S^{j}$ at this node. Using Lemma $\ref{lemma_mnK_startpattern}$ iteratively yields $C_i(S^j)\leq C_i(S^{j+1})$ for all $j\in \{0,\ldots,m-1\}$ and all $i\in N$. Iterating this argument yields $C_i(S^0) \leq C_i(S^m)$ for all $i\in N$, and in particular $C(S) = C_n(S) = C_n(S^0) \leq C_n(S^m)= C(S^m) = C(\badNash)$. 
\end{proof}

\noindent To effectively compare the social optimum with any greedy Nash Equilibrium two additional properties are required. The following technical lemma demonstrates that, under the greedy queue policy, deleting edges is never beneficial for any player.
\begin{restatable}{lemma}{lemmaAddingEdges}\label{lemma_adding edges}
Let $\Gamma_1$ and $\Gamma_2$ be packet routing games that differ only in the underlying graphs with $G(\Gamma_1)=(V,E_1)\in \mathcal{G}$ and $G(\Gamma_2)=(V,E_1\cup E_2)\in \mathcal{G}$ with $\badNash_{G_1}=\phi(\Gamma_1,0^n)$ and $\badNash_{G_2}=\phi(\Gamma_2,0^n)$. It holds that $C(\badNash_{G_1}) \geq C(\badNash_{G_2})$.
\end{restatable}
\gotoAppendix{
}{\begin{proof}
It is sufficient to show that the property holds for $|E_2\backslash E_1|=1$. Let $e=(v_j,v_{j+1})$ be the edge where the two graphs differ. Up to node $v_j$ the arrival patterns $a_{v_j}(\badNash_{G_1})$ and $a_{v_j}(\badNash_{G_2})$ trivially coincide.
Suppose that there is a player $i$ such that $a^i_{v_{j+1}}(\badNash_{G_2}) > a^i_{v_{j+1}}(\badNash_{G_1})$ and $a^{i'}_{v_{j+1}}(\badNash_{G_2}) \leq a^{i'}_{v_{j+1}}(\badNash_{G_1})$ for all $i'<i$. In particular, there are at most $|\{e\in E_1^{j+1}: \tau(e) \leq a^i_{v_{j+1}}(\badNash_{G_1}) - a^i_{v_{j}}(\badNash_{G_1})\}| -1$ players with index less than $i$ that arrive at $v_{j+1}$ at time $a^i_{v_{j+1}}(\badNash_{G_1})$. Here $E_1^{j+1}$ denotes the edges in $E_1$ that are in the $(j+1)$-th layer of $G_1$. Since $i$ arrives at $v_{j+1}$ strictly later in $G_2$ than in $G_1$, among all edges of $\{e\in E_1^{j+1}: \tau(e) \leq a^i_{v_{j+1}}(\badNash_{G_1}) - a^i_{v_{j}}(\badNash_{G_1})\}$ players arrive at $v_{j+1}$ at time $a^i_{v_{j+1}}(\badNash_{G_1})$ in $\badNash_{G_2}$. Note, that is one more player. By Lemma \ref{lemma_order} the additional player $i'$ has $i'<i$, since no overtaking is possible at any intermediate node in any equilibrium. But then we have $a^{i'}_{v_{j+1}}(\badNash_{G_2}) > a^{i'}_{v_{j+1}}(\badNash_{G_1})$, which contradicts that player $i$ was the player with the smallest index that was delayed. Therefore, $a^i_{v_{j+1}}(\badNash_{G_1}) \geq a^i_{v_{j+1}}(\badNash_{G_2})$ holds for all $i\in N$. 
In the later layers, the graphs are indistinguishable. Therefore, by iterating node by node to the destination node $d$, we obtain the statement by applying Lemma $\ref{lemma_mnK_startpattern}$ at each step.
\end{proof}}
\noindent Finally, we present our last lemma which asserts that the transit time of any player in any equilibrium in any layer with inflow at most $k$ is bounded by the $k$-th shortest transit time in that layer.
\begin{lemma}\label{lemma_inflowk}
Let $j$ be a layer with at least $k$ edges $e_1^j, e_2^j, \dots$ in a linear multigraph. Let $S \in \Nash$ be an equilibrium for a packet routing game on this graph. If the maximum inflow into $v_{j-1}$ under $S$ is at most $k$ at any point in time, the latency of a player never exceeds the $k$-th shortest transit time $\tau(e_k^j)$ in that layer.
\end{lemma}
\begin{proof}
Consider the first time $t$ in an equilibrium where a player $i$ uses an edge with a latency exceeding $\tau(e_k^j)$. Because $t$ was the earliest time with this property, we know that at time $t-1$ no player experiences a latency that exceeds $\tau(e_k^j)$. As at most $k$ players enter the queues at time $t$ and we remove one player from the queue of each used edge afterwards, either there are $k$ used edges and the latencies do not increase, or there are at most $k-1$ used edges, i.e., one of the edges $e_1^j, e_2^j, \dots, e_k^j$ is unused. In the latter case, the latency is trivially bounded by $\tau(e_k^j)$. Thus, there is still an edge with latency $\tau(e_k^j)$, which finishes the proof.
\end{proof}
\noindent With Proposition~\ref{prop_together} and all the lemmata at hand, we are now able to prove an upper bound of 2 on the $\PoA$ for the class of linear multigraphs.

\begin{theorem}\label{satz_PoA2} For the family $\mathcal{G}$ of linear multigraphs, we have
$\PoA(\mathcal{G})\leq 2$.
\end{theorem}
\begin{proof}
For a fixed packet routing game $\Gamma$, there is a $k$ such that there exists an optimal flow $S^*$ that is essentially temporally repeated and uses $k$ disjoint paths due to Theorem $\ref{thm:opt}$. Consider the subgraph \mbox{$G'=(V,E')$} of $G=(V,E)$, which consists only of all edges used by this optimal flow. Consider the induced packet routing game $\Gamma'$ from $\Gamma$ that differs only on the underlying graph, i.e., \mbox{$G(\Gamma')=G'$} and $G(\Gamma)=G$. Clearly it holds that $C^{\Gamma}(S^*)= C^{\Gamma'}(S^*)$. \mbox{Proposition \ref{prop_together}} and Lemma \ref{lemma_adding edges} show that for the worst equilibria $\badNash_{G}=\phi(\Gamma,0^n)$ and $\badNash_{G'}=\phi(\Gamma',0^n)$ in the two graphs we have $C^{\Gamma}(\badNash_{G}) \leq C^{\Gamma'}(\badNash_{G'})$ and, hence, $\PoA(\Gamma)\leq \PoA(\Gamma')$. Thus, it suffices to show the claim for the subgraph $G'$.

If $G'$ has only one layer, the definitions of the equilbria and the social optimum coincide as there are no intermediate nodes. Thus, $S^*$ has the same arrival pattern at $d$ as the equilibrium $\badNash_{G'}$ and hence $\PoA(\Gamma')=1$. 

If $G'$ consists of more than one layer, we will separately discuss the latency of the last player from $s$ to $v_1$ and from $v_1$ to $d$, respectively. In any UFR equilibrium $S$, every player aims to be as fast as possible at node $v_1$. Thus, for each time $t$, a player arrives at $v_1$ via all edges $e^1$ with $\tau(e^1) \leq t$. It is immediate that for each $t$, the equilibrium $S$ maximizes the number of players arriving at $v_1$ by $t$. Additionally, the players arrive at $v_1$ in the order of their indices due to Lemma~\ref{lemma_order}. Thus, the $n$-th player reaches $v_1$ in $S$ not later than the last player reaches $v_1$ in $S^*$, and thus we have in particular $a^n_{v_1}(S)\leq \max_{i\in N}a^i_{v_1}(S^*) \leq C^{\Gamma'}(S^*)$. 
For the second part from $v_1$ to $d$, note that the inflow in each layer except the first is at most $k$, since $G'$ is constructed by $k$ edge-disjoint $s-d$ paths.
Therefore, by Lemma \ref{lemma_inflowk}, a player in an equilibrium has a maximum latency of $\tau(e_k^j)$ in any layer $j>1$. We conclude that for each player in an equilibrium, we can bound the total latency from $v_1$ to $d$ from above by $\sum_{j=2}^m \tau(e_k^j)$. 
As seen in Theorem $\ref{thm:opt}$, we have $\sum_{j=1}^m \tau(e_k^j)\leq C^{\Gamma'}(S^*)$. Therefore, we can further bound the total latency from $v_1$ to $d$ from above by $\sum\limits_{j=2}^m \tau(e_k^j)\leq C^{\Gamma'}(S^*)-\tau(e_k^1)$.\\
Adding both parts, we obtain
\begin{align*}
    C^{\Gamma'}(\badNash_{G'})&\leq a^n_{v_1}(S) + \left(C^{\Gamma'}(S^*) -\tau(e_k^1)\right) \leq 2 \cdot C^{\Gamma'}(S^*).
\end{align*}
Thus, for all $\Gamma$ with $G(\Gamma) \in \mathcal{G}$ it follows \[\PoA(\Gamma)\leq \PoA(\Gamma')=\frac{C^{\Gamma'}(\badNash_{G'})}{C^{\Gamma'}(S^*)} \leq \frac{2 \cdot C^{\Gamma'}(S^*)}{C^{\Gamma'}(S^*)} \leq 2\;.\]
This yields $\PoA(\mathcal{G}) = \sup\limits_{\Gamma: G(\Gamma) \in \mathcal{G}} \PoA(\Gamma) \leq 2\;.$
\end{proof}
\section{Lower Bound on the Price of Stability}\label{sec:pos}

\subsection{Constructing Bad Instances}\label{sec:posBadInst}
In this section, we will establish a lower bound on the $\PoS$ and therefore also on the $\PoA$. We will construct a sequence of packet routing games $(\Gamma_i)_{i \in \N}$ with $\PoA(\Gamma_i) = \PoS(\Gamma_i)$ and $\lim\limits_{i \rightarrow \infty} \PoS(\Gamma_i) \geq \frac{e}{e-1}$, where each $G(\Gamma_i)\in \mathcal{G}$, i.e., each underlying graph is a linear multigraph. This lower bound is slightly above 1.582.

We introduce a subclass $\mathcal{G}_{k,\ell}\subseteq \mathcal{G}$ for $0\leq\ell<k$, $\ell \in \N_{\geq 0}$ and $k\in \N_{>0}$. A graph $G\in\mathcal{G}_{k,\ell}$ is a linear multigraph that consists of nodes $V=\{s=v_0,$ $v_1,\ldots,v_{k-\ell}=d\}$ and $k$ edges in every layer. In layer $j$ from $v_{j-1}$ to $v_j$, these $k$ edges consist of $j-1$ edges with a layer specific transit time of $\tau_j$ (called \emph{special edges}) and $k-(j-1)$ edges with a transit time of $1$ (called \emph{standard edges}). In Figure \ref{klayersl1}, the graphs are visualized for different values of $\ell$. 

An informal interpretation of the graph class $\mathcal{G}_{k,\ell}$ is at hand. The value $k$ clearly represents the number of edges in each layer of the graph. Furthermore, it also determines the maximum number of nodes in the graph, which is $k+1$. 
On the other hand, $\ell$ signifies the number of layers that are truncated from the graph. Thus, a graph in $\mathcal{G}_{k,\ell}$ has $k-\ell+1$ nodes.
Additionally, the amount of special edges incrementally increases by one as we move from one layer to the next, starting with zero special edges in the first layer.

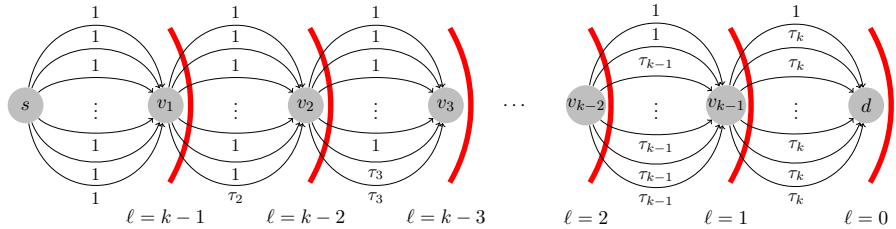
\begin{figure}[H]
    \resizebox{\linewidth}{!}{\tikzstyle{vertex}=[circle,fill=black!25,minimum size=20pt,inner sep=0pt]
\tikzstyle{novertex}=[]
\tikzstyle{smallvertex}=[circle,fill=black,minimum size=6pt,inner sep=0pt]
\tikzstyle{edge} = [draw,thin,->]
\tikzstyle{weight} = [font=\small]
\tikzstyle{selected edge} = [draw,line width=3pt,-,red!50]
\tikzstyle{big edge} = [draw,line width=2pt,->,black]
\tikzstyle{cutting edge} = [draw,line width=3pt,-,red]
\begin{tikzpicture}[scale=0.55,auto,swap]
    \node[vertex](s1) at (0,0){$s$};
    \node[vertex](s2) at (5,0){$v_1$};
    \node[novertex](s2o) at (5,3){};
    \node[novertex](s2u) at (5,-3){};
    \node[novertex](s2t) at (5,-4){$\ell=k-1$};

    \node[vertex](s3) at (10,0){$v_2$};
    \node[novertex](s3o) at (10,3){};
    \node[novertex](s3u) at (10,-3){};
    \node[novertex](s3t) at (10,-4){$\ell=k-2$};
    
    \node[vertex](s35) at (15,0){$v_3$};
    \node[novertex](s35o) at (15,3){};
    \node[novertex](s35u) at (15,-3){};
    \node[novertex](s35t) at (15,-4){$\ell=k-3$};
    
    \node[vertex](s4) at (20,0){$v_{k-2}$};
    \node[novertex](s4o) at (20,3){};
    \node[novertex](s4u) at (20,-3){};
    \node[novertex](s4t) at (20,-4){$\ell=2$};
    
    \node[vertex](s5) at (25,0){$v_{k-1}$};
    \node[novertex](s5o) at (25,3){};
    \node[novertex](s5u) at (25,-3){};
    \node[novertex](s5t) at (25,-4){$\ell=1$};
    
    \node[vertex](t) at (30,0){$d$};
    \node[novertex](to) at (30,3){};
    \node[novertex](tu) at (30,-3){};
    \node[novertex](tt) at (30,-4){$\ell=0$};

    \draw [edge] (s1) to[out=80,in=100, distance=3cm ] node[weight,above,black]{$1$} (s2);
    \draw [edge] (s1) to[out=70,in=110, distance=2cm ] node[weight,above,black]{$1$} (s2);
    \draw [edge] (s1) to[out=45,in=135, distance=1cm ] node[weight,above,black]{$1$} (s2);
    \node[text width=0.1cm] at (2.5,0) {$\vdots$};
    \draw [edge] (s1) to[out=-45,in=-135, distance=1cm ] node[weight,below,black]{$1$} (s2);
    \draw [edge] (s1) to[out=-70,in=-110, distance=2cm ] node[weight,below,black]{$1$} (s2);
    \draw [edge] (s1) to[out=-80,in=-100, distance=3cm ] node[weight,below,black]{$1$} (s2);

    \draw [cutting edge] (s2o) to[out=-60,in=60, distance=2cm ] node[weight,below,black]{} (s2u);    

    \draw [edge] (s2) to[out=80,in=100, distance=3cm ] node[weight,above,black]{$1$} (s3);
    \draw [edge] (s2) to[out=70,in=110, distance=2cm ] node[weight,above,black]{$1$} (s3);
    \draw [edge] (s2) to[out=45,in=135, distance=1cm ] node[weight,above,black]{$1$} (s3);
    \node[text width=0.1cm] at (7.5,0) {$\vdots$};
    \draw [edge] (s2) to[out=-45,in=-135, distance=1cm ] node[weight,below,black]{$1$} (s3);
    \draw [edge] (s2) to[out=-70,in=-110, distance=2cm ] node[weight,below,black]{$1$} (s3);
    \draw [edge] (s2) to[out=-80,in=-100, distance=3cm ] node[weight,below,black]{$\tau_2$} (s3);

    \draw [cutting edge] (s3o) to[out=-60,in=60, distance=2cm ] node[weight,below,black]{} (s3u); 

    \draw [edge] (s3) to[out=80,in=100, distance=3cm ] node[weight,above,black]{$1$} (s35);
    \draw [edge] (s3) to[out=70,in=110, distance=2cm ] node[weight,above,black]{$1$} (s35);
    \draw [edge] (s3) to[out=45,in=135, distance=1cm ] node[weight,above,black]{$1$} (s35);
    \node[text width=0.1cm] at (12.5,0) {$\vdots$};
    \draw [edge] (s3) to[out=-45,in=-135, distance=1cm ] node[weight,below,black]{$1$} (s35);
    \draw [edge] (s3) to[out=-70,in=-110, distance=2cm ] node[weight,below,black]{$\tau_3$} (s35);
    \draw [edge] (s3) to[out=-80,in=-100, distance=3cm ] node[weight,below,black]{$\tau_3$} (s35);

    \draw [cutting edge] (s35o) to[out=-60,in=60, distance=2cm ] node[weight,below,black]{} (s35u);
    
    \node[] at (17.5,0) {$\cdots$};

    \draw [cutting edge] (s4o) to[out=-60,in=60, distance=2cm ] node[weight,below,black]{} (s4u);
    
    \draw [edge] (s4) to[out=80,in=100, distance=3cm ] node[weight,above,black]{$1$} (s5);
    \draw [edge] (s4) to[out=70,in=110, distance=2cm ] node[weight,above,black]{$1$} (s5);
    \draw [edge] (s4) to[out=45,in=135, distance=1cm ] node[weight,above,black]{$\tau_{k-1}$} (s5);
    \node[text width=0.1cm] at (22.5,0) {$\vdots$};
    \draw [edge] (s4) to[out=-45,in=-135, distance=1cm ] node[weight,below,black]{$\tau_{k-1}$} (s5);
    \draw [edge] (s4) to[out=-70,in=-110, distance=2cm ] node[weight,below,black]{$\tau_{k-1}$} (s5);
    \draw [edge] (s4) to[out=-80,in=-100, distance=3cm ] node[weight,below,black]{$\tau_{k-1}$} (s5);

    \draw [cutting edge] (s5o) to[out=-60,in=60, distance=2cm ] node[weight,below,black]{} (s5u);    

    \draw [edge] (s5) to[out=80,in=100, distance=3cm ] node[weight,above,black]{$1$} (t);
    \draw [edge] (s5) to[out=70,in=110, distance=2cm ] node[weight,above,black]{$\tau_k$} (t);
    \draw [edge] (s5) to[out=45,in=135, distance=1cm ] node[weight,above,black]{$\tau_k$} (t);
    \node[text width=0.1cm] at (27.5,0) {$\vdots$};
    \draw [edge] (s5) to[out=-45,in=-135, distance=1cm ] node[weight,below,black]{$\tau_k$} (t);
    \draw [edge] (s5) to[out=-70,in=-110, distance=2cm ] node[weight,below,black]{$\tau_k$} (t);
    \draw [edge] (s5) to[out=-80,in=-100, distance=3cm ] node[weight,below,black]{$\tau_k$} (t);

    \draw [cutting edge] (to) to[out=-60,in=60, distance=2cm ] node[weight,below,black]{} (tu);
    
 \end{tikzpicture}

 } 
    \caption{Visualization of graphs in $\mathcal{G}_{k,\ell}$. For a given $\ell$, the graph $G$ consists of all layers left of the red curved line and the last node before the line serves as destination node $d$.}
    \label{klayersl1}
\end{figure}

\begin{restatable}{theorem}{satzPoAgeqee}\label{satz_PoAgeqee1}
There exists a sequence of packet routing games $(\Gamma_i)_{i \in \N_{>0}}$ on linear multigraphs with
$\lim\limits_{i \rightarrow \infty} \PoS(\Gamma_i) \geq \frac{e}{e-1}$ where each $G(\Gamma_i)\in \mathcal{G}$. Thus, we obtain $\PoS(\mathcal{G})\geq \frac{e}{e-1}\geq 1.582$.
\end{restatable}
\begin{proof}
    For each $i\in \N_{>0}$, we choose $\Gamma_i$ to be the game characterized by $G(\Gamma_i)\in \mathcal{G}_{k,\ell}$ with $\ell = i$ and $k=\lceil ei\rceil$. The player set is $N(\Gamma_i)=[n_i]$ with $n_i= k!= \lceil ei\rceil!$ and the transit time of the special edges in layer $j\geq 2$ is set to
\begin{align*}
    \tau_j = \frac{n_i}{k-j+1} - \frac{n_i}{k-j+2} +2.
\end{align*}
The underlying concept of this construction can be explained as follows. The transit times on the special edges are deliberately set high to ensure that no equilibrium flow uses them. As a result, the network progressively narrows with each layer, leading to inevitable queuing on all standard edges. Note that the transit times of the special edges increase for higher layers. Thus, we additionally truncate the construction at a certain point to prevent the narrow layers from significantly enlarging the social optimum too much.
We begin by showing by induction on the arrival times of players at intermediate nodes that no special edges are used in any equilibrium. Specifically, we show that for $j\geq 1$ at each time step $j \leq t \leq \frac{n_i}{k-j+1}+j-1$ exactly $k-j+1$ players arrive at $v_j$. 
Consider a fixed game $\Gamma_i$ and an arbitrary equilibrium $S_i\in \Nash^{\Gamma_i}$.
First, as the first layer is solely composed of $k$ standard edges, trivially no player uses a special edge and it is immediate that $k$ players arrive at $v_1$ at each time $1\leq t \leq \frac{n_i}{k}$. 
Suppose by induction hypothesis that this holds up to layer $j$, i.e., $k-j+1$ players arrive at time steps $j \leq t \leq \frac{n_i}{k-j+1}+j-1$ at $v_j$. We have $k-j$ standard edges in layer $j+1$, hence, the throughput on standard edges is one smaller than the inflow. Therefore, in each time step, the number of players in queues in this layer increases by 1. After $\frac{n_i}{k-j+1}$ time steps, all $n_i$ players arrived at $v_j$. At this point in time, there are $\frac{n_i}{k-j+1}$ queued players in $k-j$ many queues, that is, each queue has length exactly $\frac{n_i}{(k-j+1)(k-j)}=\frac{n_i}{k-j}-\frac{n_i}{k-j+1}$. Since $\tau_{j+1}$ is strictly larger than this value, the special edges in layer $j+1$ are indeed never used. Furthermore, $k-j=k-(j+1)+1$ players arrive at each time step $j+1 \leq t \leq \frac{n_i}{k-j+2}+j$ at $v_{j+1}$.

As no special edges are used by any equilibrium $S_i\in\Nash$, this implies that every equilibrium has the same completion time. The $n_i$-th player arrives at node $v_j$ at time $\frac{n_i}{k-j+1}+ j-1$ and thus we have
\begin{align}
    C^{\Gamma_i}(S_i) &= (k-\ell) + \frac{n_i}{k- (k-\ell)+1} -1 = (k-\ell -1) + n_i\cdot \left(\frac{1}{\ell+1}\right).\label{eq:nash}
\end{align}
It remains to compute the completion time for the social optimum. 
Let $P^1$ denote a shortest $s-d$ path in $G$ and $P^j$ a shortest $s-d$ path in $G$ after deleting paths $P^1, \dots, P^{j-1}$. According to Section \ref{sec:understandingopt}, there is an optimal state $S^*_i$ of $\Gamma_i$ of the following form. $C^{\Gamma_i}(S^*_i)+1-\tau(P^1)$ packets use $P^1$ and $C^{\Gamma_i}(S^*_i)+\delta_j -\tau(P^j)$ packets use $P^j$ with $\delta_j\in\{0,1\}$ for \mbox{$2\leq j \leq k'\leq k$} with \mbox{$\sum_{j=2}^{k'} \delta_j$}$ = $\mbox{$n_i-k'\cdot C^{\Gamma_i}(S^*_i) -1 +\sum_{j=1}^{k'} \tau(P^j)$}. 
Together with the fact that the optimal state does not use paths $P^{k'+1}, \dots, P^k$, i.e., $\tau(P^k)\geq \dots \geq \tau(P^{k'+1}) > C^{\Gamma_i}(S^*_i)$, we obtain that the total number of packets sent by $S^*_i$ is
\begin{align*}
    n_i &= C^{\Gamma_i}(S^*_i)+1-\tau(P^1) + \sum_{j=2}^{k'} \left(C^{\Gamma_i}(S^*_i)+\delta_j -\tau(P^j)\right)\nonumber\\
    &= k'\cdot C^{\Gamma_i}(S^*_i)+1-\sum_{j=1}^{k'}\tau(P^j) + \sum_{j=2}^{k'} \delta_j \\
    &\geq k'\cdot C^{\Gamma_i}(S^*_i)-\sum_{j=1}^{k'}\tau(P^j)\nonumber\\
    &\geq k\cdot C^{\Gamma_i}(S^*_i)-\sum_{j=1}^{k}\tau(P^j)  \\
    &= k \cdot C^{\Gamma_i}(S^*_i)-  \Biggl(\underbrace{\sum\limits_{j=1}^k j - \sum\limits_{j=1}^{\ell} j }_{\text{standard edges}}+ \underbrace{\sum\limits_{j=2}^{k-\ell} (j-1)\tau_j}_{\text{special edges}}\Biggr)\nonumber\\
    &= k\cdot C^{\Gamma_i}(S^*_i) - \left(\frac{k^2+k}{2} - \frac{\ell^2+\ell}{2} + \sum\limits_{j=2}^{k-\ell} (j-1)\tau_j\right).
\end{align*}
Rearranging terms implies an upper bound on $C^{\Gamma_i}(S^*_i)$.
\[C^{\Gamma_i}(S^*_i)\leq \frac{1}{k}\left(n_i+\frac{k^2+k}{2} - \frac{\ell^2+\ell}{2} +\sum\limits_{j=2}^{k-\ell} (j-1)\tau_j\right)\]
With the help of the telescope sum 
\begin{align*}
    &\sum\limits_{j=2}^{k-\ell} (j-1)\left( \frac{n}{k-j+1} -  \frac{n}{k-j+2}+2\right) \nonumber\\
    = & (k-\ell)(k-\ell-1) + (k-\ell-1)\left( \frac{n}{\ell+1}\right) - \sum_{j=0}^{k-\ell-2} \frac{n}{k-j}
\end{align*}
we get
\begin{align}
    C^{\Gamma_i}&(S^*_i) \leq \frac{1}{k}\left(n_i+\frac{k^2+k}{2} - \frac{\ell^2+\ell}{2} + (k-\ell)(k-\ell-1) \right.\nonumber\\
    &\hspace{1.7cm}\left.+ (k-\ell-1)\left( \frac{n_i}{\ell+1}\right) - \sum\limits_{j=0}^{k-\ell-2} \frac{n_i}{k-j}\right)\nonumber\\
    &= \left( \frac{3k^2-4k\ell-k+\ell^2+\ell}{2k}\right) + n_i \cdot \left( \frac{k-\ell-1}{(\ell+1)k} + \frac{1}{k} -\frac{1}{k}\sum\limits_{j=\ell+2}^{k} \frac{1}{j}\right).\label{eq:optUB}
\end{align}
For the j-th harmonic number $H_j$, it is well known that $H_j= \ln(j) + \gamma + \frac{1}{2j} - \varepsilon_j$ with $0\leq \varepsilon_j \leq \frac{1}{8j^2}$ and the Euler-Mascheroni constant $\gamma \approx 0.577$. Thus,
\begin{align}
    \lim\limits_{\ell \rightarrow \infty} \sum\limits_{j=\ell+2}^{\left\lceil e \ell\right\rceil} \frac{1}{j} 
    &= \lim\limits_{\ell \rightarrow \infty} \left(H_{\left\lceil e \ell\right\rceil}-H_{\ell+1}\right)\nonumber \\
    &=  \lim\limits_{\ell \rightarrow \infty} \left(\underbrace{\ln\left(\frac{\left\lceil e \ell\right\rceil}{\ell+1}\right)}_{\rightarrow \ln(e) = 1} + \underbrace{\frac{1}{2\left\lceil e \ell\right\rceil} - \varepsilon_{\left\lceil e \ell\right\rceil} - \frac{1}{2(\ell+1)} + \varepsilon_{\ell+1}}_{\rightarrow 0}\right) = 1 .\label{eq:harmonic}
\end{align}
Since the growth of $n_i$ dominates the term, we obtain 
\begin{align*}
    \PoS(\mathcal{G}) &\geq \lim\limits_{i \rightarrow \infty} \PoS(\Gamma_i) =\lim\limits_{i \rightarrow \infty} \frac{C^{\Gamma_i}(S_i)}{C^{\Gamma_i}(S^*_i)}\\
    &\hspace{-0.25cm}\overset{(\ref{eq:nash}), (\ref{eq:optUB}) }{\geq }\lim\limits_{i \rightarrow \infty} \frac{(k-\ell -1) + n_i\cdot \left(\frac{1}{\ell+1}\right)}{ \left( \frac{3k^2-4k\ell-k+\ell^2+\ell}{2k}\right) + n_i \cdot \left( \frac{k-\ell-1}{(\ell+1)k} + \frac{1}{k} -\frac{1}{k}\sum\limits_{j=\ell+2}^{k} \frac{1}{j}\right)}.
\end{align*}
By simplifying this expression, we get
\begin{align*}
    \PoS(\mathcal{G})&\geq \lim\limits_{i  \rightarrow \infty} \frac{\frac{1}{\ell+1}}{ \frac{k-\ell-1}{(\ell+1)k} + \frac{1}{k} -\frac{1}{k}\sum\limits_{j=\ell+2}^{k} \frac{1}{j}} 
    = \lim\limits_{i \rightarrow \infty} \frac{1}{ 1 -\frac{\ell+1}{k} \sum\limits_{j=\ell+2}^{k} \frac{1}{j}} \\
    &=  \lim\limits_{i \rightarrow \infty} \frac{1}{ 1 -\underbrace{\frac{\ell+1}{\left\lceil e \ell\right\rceil}}_{\rightarrow \frac{1}{e}} \underbrace{\sum\limits_{j=\ell+2}^{\left\lceil e \ell\right\rceil} \frac{1}{j}}_{\rightarrow 1, \text{ due to } (\ref{eq:harmonic})}} 
    = \frac{1}{1-\frac{1}{e}\cdot 1}
    = \frac{1}{\left(\frac{e-1}{e}\right)}
    = \frac{e}{e-1}\;,
\end{align*}
which finishes the proof.
\end{proof}

\subsection{Implications for Flows over Time}\label{sec:FlowOverTime}
Flows over time can essentially be seen as the continuous variant of packet routings. For a comprehensive introduction to the topic, we refer to the survey article of Skutella~\cite{Skutella2009survey}. 
We will follow its lines and recall the most basic definitions in Appendix \ref{app:FlowOverTime}.

Correa et al.\ proved an upper bound on the makespan-$\PoA$ of Nash flows over time of $\frac{e}{e-1}$ if a \emph{monotonicity conjecture} holds~\cite{DBLP:journals/mor/CorreaCO22}. In essence, the monotonicity conjecture states that the completion time of a Nash flow over time does not decrease as the inflow rate increases while the total amount of flow and the underlying network remains the same.
\begin{proposition}[Correa et al.~\cite{DBLP:journals/mor/CorreaCO22}]
    The makespan-$\PoA$ is upper bounded by $\frac{e}{e-1}$ if the monotonicity conjecture holds.
\end{proposition}

\noindent Correa et al.~\cite{DBLP:journals/mor/CorreaCO22} proved the monotonicity conjecture for linear multigraphs. Interestingly, the same authors also provided a lower bound on the $\PoA$ of $\frac{e}{e-1}$, by a sequence of instances for which the monotonicity conjecture could not been proven. We tighten the $\PoA$ bound on linear multigraphs by observing that the instance used in the proof of Theorem \ref{satz_PoAgeqee1} can be extended to Nash flows over time. The proof exploits the fact that the instance is constructed with evenly loaded edges, such that interpreting integral flow values continuously does not allow for any improvement in the travel time for any flow particle.

\begin{restatable}{proposition}{propositionFlowOverTimeLowerBound}\label{prop_propositionFlowOverTimeLowerBound}
    The makespan-$\PoA$ is lower bounded by $\frac{e}{e-1}$ even in linear multigraphs.
\end{restatable}
\section{Conclusion}
We provided an upper bound of $2$ on the $\PoA$ in linear multigraphs for one of the most natural packet routing games. Interestingly, our upper bound is independent of the network size and the number of players which is in stark contrast to other packet routing games. Furthermore, a sequence of linear multigraphs with a $\PoS$ that converges to at least $\frac{e}{e-1}$ has been presented. 
It is evident that our findings extend to serial concatenations of parallel path networks through subdivisions of edges in linear multigraphs and we obtain the same bounds on the $\PoA$ in an extended model where we allow for integer-valued edge capacities, as shown in Appendix \ref{sec:edgcap}. 

A natural next step would be to further tighten the bounds and establish an upper bound on the PoA for more general graph classes. Furthermore, considering the PoA when not restricted to UFR equilibria is an interesting research direction. Future research could explore the PoA for other relevant social objective functions as well.



%
%
%
\bibliographystyle{plainnat}
\bibliography{references}

\begin{thebibliography}{24}
\providecommand{\natexlab}[1]{#1}
\providecommand{\url}[1]{\texttt{#1}}
\expandafter\ifx\csname urlstyle\endcsname\relax
  \providecommand{\doi}[1]{doi: #1}\else
  \providecommand{\doi}{doi: \begingroup \urlstyle{rm}\Url}\fi

\bibitem[Braess(1968)]{braess1968paradoxon}
Dietrich Braess.
\newblock {\"U}ber ein {P}aradoxon aus der {V}erkehrsplanung.
\newblock \emph{Unternehmensforschung}, 12:\penalty0 258--268, 1968.
\newblock URL \url{https://doi.org/10.1007/BF01918335}.

\bibitem[Cao et~al.(2017)Cao, Chen, Chen, and Wang]{DBLP:conf/sigecom/CaoCCW17}
Zhigang Cao, Bo~Chen, Xujin Chen, and Changjun Wang.
\newblock A network game of dynamic traffic.
\newblock In Constantinos Daskalakis, Moshe Babaioff, and Herv{\'{e}} Moulin,
  editors, \emph{Proceedings of the 2017 {ACM} Conference on Economics and
  Computation, {EC} '17, Cambridge, MA, USA, June 26-30, 2017}, pages 695--696.
  {ACM}, 2017.
\newblock URL \url{https://doi.org/10.1145/3033274.3085101}.

\bibitem[Cao et~al.(2022)Cao, Chen, Chen, and Wang]{cao2022bounding}
Zhigang Cao, Bo~Chen, Xujin Chen, and Changjun Wang.
\newblock Bounding residence times for atomic dynamic routings.
\newblock \emph{Mathematics of Operations Research}, 47\penalty0 (4):\penalty0
  3261--3281, 2022.
\newblock URL \url{https://doi.org/10.1287/moor.2021.1242}.

\bibitem[Cominetti et~al.(2011)Cominetti, Correa, and
  Larr{\'e}]{cominetti2011existence}
Roberto Cominetti, Jos{\'e}~R Correa, and Omar Larr{\'e}.
\newblock Existence and uniqueness of equilibria for flows over time.
\newblock In \emph{International Colloquium on Automata, Languages, and
  Programming}, pages 552--563. Springer, 2011.
\newblock URL \url{https://doi.org/10.1007/978-3-642-22012-8\_44}.

\bibitem[Correa et~al.(2022)Correa, Cristi, and
  Oosterwijk]{DBLP:journals/mor/CorreaCO22}
Jos{\'{e}} Correa, Andr{\'{e}}s Cristi, and Tim Oosterwijk.
\newblock On the price of anarchy for flows over time.
\newblock \emph{Mathematics of Operations Research}, 47\penalty0 (2):\penalty0
  1394--1411, 2022.
\newblock URL \url{https://doi.org/10.1287/moor.2021.1173}.

\bibitem[Correa et~al.(2004)Correa, Schulz, and Stier-Moses]{correa2004selfish}
Jos{\'e}~R Correa, Andreas~S Schulz, and Nicol{\'a}s~E Stier-Moses.
\newblock Selfish routing in capacitated networks.
\newblock \emph{Mathematics of Operations Research}, 29\penalty0 (4):\penalty0
  961--976, 2004.
\newblock URL \url{https://doi.org/10.1287/moor.1040.0098}.

\bibitem[Fleischer and Tardos(1998)]{DBLP:journals/orl/FleischerT98}
Lisa Fleischer and {\'{E}}va Tardos.
\newblock Efficient continuous-time dynamic network flow algorithms.
\newblock \emph{Oper. Res. Lett.}, 23\penalty0 (3-5):\penalty0 71--80, 1998.
\newblock URL \url{https://doi.org/10.1016/S0167-6377(98)00037-6}.

\bibitem[Ford~Jr and Fulkerson(1958)]{ford1958constructing}
Lester~R Ford~Jr and Delbert~Ray Fulkerson.
\newblock Constructing maximal dynamic flows from static flows.
\newblock \emph{Operations research}, 6\penalty0 (3):\penalty0 419--433, 1958.

\bibitem[Han et~al.(2013)Han, Friesz, and Yao]{HAN201317}
Ke~Han, Terry~L. Friesz, and Tao Yao.
\newblock Existence of simultaneous route and departure choice dynamic user
  equilibrium.
\newblock \emph{Transportation Research Part B: Methodological}, 53:\penalty0
  17--30, 2013.
\newblock ISSN 0191-2615.
\newblock URL \url{https://doi.org/10.1016/j.trb.2013.01.009}.

\bibitem[Harks et~al.(2018)Harks, Peis, Schmand, Tauer, and
  Vargas~Koch]{harks2018competitive}
Tobias Harks, Britta Peis, Daniel Schmand, Bjoern Tauer, and Laura Vargas~Koch.
\newblock Competitive packet routing with priority lists.
\newblock \emph{ACM Transactions on Economics and Computation (TEAC)},
  6\penalty0 (1):\penalty0 1--26, 2018.
\newblock URL \url{https://doi.org/10.1145/3184137}.

\bibitem[Hoefer et~al.(2011)Hoefer, Mirrokni, R{\"o}glin, and
  Teng]{hoefer2011competitive}
Martin Hoefer, Vahab~S Mirrokni, Heiko R{\"o}glin, and Shang-Hua Teng.
\newblock Competitive routing over time.
\newblock \emph{Theoretical Computer Science}, 412\penalty0 (39):\penalty0
  5420--5432, 2011.
\newblock URL \url{https://doi.org/10.1016/j.tcs.2011.05.055}.

\bibitem[Koch and Skutella(2011)]{kochskutella2011}
Ronald Koch and Martin Skutella.
\newblock Nash equilibria and the price of anarchy for flows over time.
\newblock \emph{Theory of Computing Systems}, 49\penalty0 (1):\penalty0 71--97,
  2011.

\bibitem[Olver et~al.(2021)Olver, Sering, and Koch]{DBLP:conf/focs/OlverSK21}
Neil Olver, Leon Sering, and Laura~Vargas Koch.
\newblock Continuity, uniqueness and long-term behavior of nash flows over
  time.
\newblock In \emph{62nd {IEEE} Annual Symposium on Foundations of Computer
  Science, {FOCS} 2021, Denver, CO, USA, February 7-10, 2022}, pages 851--860.
  {IEEE}, 2021.
\newblock URL \url{https://doi.org/10.1109/FOCS52979.2021.00087}.

\bibitem[Pigou(1920)]{pigou}
Arthur Pigou.
\newblock \emph{The economics of welfare}.
\newblock Macmillan, 1920.

\bibitem[Rosenthal(1973)]{rosenthal1973network}
Robert~W Rosenthal.
\newblock The network equilibrium problem in integers.
\newblock \emph{Networks}, 3\penalty0 (1):\penalty0 53--59, 1973.
\newblock URL \url{https://doi.org/10.1002/net.3230030104}.

\bibitem[Roughgarden(2005)]{roughgarden2005selfish}
Tim Roughgarden.
\newblock \emph{Selfish routing and the price of anarchy}.
\newblock MIT press, 2005.

\bibitem[Scarsini et~al.(2018)Scarsini, Schr{\"{o}}der, and
  Tomala]{DBLP:journals/ior/ScarsiniST18}
Marco Scarsini, Marc Schr{\"{o}}der, and Tristan Tomala.
\newblock Dynamic atomic congestion games with seasonal flows.
\newblock \emph{Operations Research}, 66\penalty0 (2):\penalty0 327--339, 2018.

\bibitem[Scheffler et~al.(2022)Scheffler, Strehler, and
  Vargas~Koch]{scheffler2022routing}
Robert Scheffler, Martin Strehler, and Laura Vargas~Koch.
\newblock Routing games with edge priorities.
\newblock \emph{ACM Transactions on Economics and Computation}, 10\penalty0
  (1):\penalty0 1--27, 2022.
\newblock URL \url{https://doi.org/10.1145/3488268}.

\bibitem[Sering(2020)]{PhDLeon}
Leon Sering.
\newblock \emph{Nash flows over time}.
\newblock Dissertation, {TU} Berlin, Berlin, 2020.
\newblock URL \url{https://doi.org/10.14279/depositonce-10640}.

\bibitem[Skutella(2008)]{DBLP:conf/bonnco/Skutella08}
Martin Skutella.
\newblock An introduction to network flows over time.
\newblock In William~J. Cook, L{\'{a}}szl{\'{o}} Lov{\'{a}}sz, and Jens Vygen,
  editors, \emph{Research Trends in Combinatorial Optimization}, pages
  451--482. Springer, 2008.
\newblock URL \url{https://doi.org/10.1007/978-3-540-76796-1\_21}.

\bibitem[Skutella(2009)]{Skutella2009survey}
Martin Skutella.
\newblock \emph{An Introduction to Network Flows over Time}, pages 451--482.
\newblock Springer Berlin Heidelberg, Berlin, Heidelberg, 2009.
\newblock URL \url{https://doi.org/10.1007/978-3-540-76796-1_21}.

\bibitem[Vargas~Koch(2020)]{PhDLaura}
Laura Vargas~Koch.
\newblock \emph{{C}ompetitive variants of discrete and continuous flows over
  time}.
\newblock Dissertation, RWTH Aachen University, Aachen, 2020.
\newblock URL \url{https://doi.org/10.18154/RWTH-2020-11648}.

\bibitem[Vickrey(1969)]{vickrey1969congestion}
William~S Vickrey.
\newblock Congestion theory and transport investment.
\newblock \emph{The American economic review}, 59\penalty0 (2):\penalty0
  251--260, 1969.

\bibitem[Werth et~al.(2014)Werth, Holzhauser, and Krumke]{WERTH201418}
T.L. Werth, M.~Holzhauser, and S.O. Krumke.
\newblock Atomic routing in a deterministic queuing model.
\newblock \emph{Operations Research Perspectives}, 1\penalty0 (1):\penalty0
  18--41, 2014.
\newblock ISSN 2214-7160.
\newblock URL \url{https://doi.org/10.1016/j.orp.2014.05.001}.

\end{thebibliography}

\ifappendix
\newpage
\appendix
\section{Appendix}\label{app}
\subsection{Choice of Equilibria}\label{sec:equi}
When dropping the assumption that a player needs to be as fast as possible at every intermediate node and instead only needs to arrive as early as possible at the destination, strange behavior patterns can arise that do not necessarily mirror the behavior of road traffic participants in reality. In particular, there might be instances where players with a higher index can arrive at intermediate nodes strictly earlier than players with a lower index and therefore overtake a player with higher priority for some time. Only due to the tie-breaking rules, they will later be overtaken again.  
Consider the following graph, visualized in Figure \ref{fig:overtakingInNE}, with the strategy profile
$P_{(1)}=(e_1^1,e_1^2,e_1^3), P_{(2)}=(e_2^1,e_2^2,e_1^3)$, $P_{(3)}=(e_3^1,e_1^2,e_1^3)$, $P_{(4)}=(e_1^1,e_2^2,e_1^3)$, $P_{(5)}=(e_2^1,e_1^2,e_1^3)$, $P_{(6)}=(e_3^1,e_2^2,e_1^3)$, $P_{(7)}=(e_1^1,e_1^2,e_1^3)$, \mbox{$P_{(8)}=(\mathbf{e_4^1},e_2^2,e_1^3)$}, $P_{(9)}=(e_2^1,e_2^2,e_1^3)$.
Clearly, due to all players arriving at $t$ in order of their index one at a time beginning at time $3$ all players arrive as early as possible at the destination. Note that player 8 takes a detour by taking the longer edge $e_4^1$ and is strictly later at node $v_1$ than player 9, while arriving at $d$ strictly before player $8$. Hence this state is not an UFR equilibrium but all players arrive as early as possible at the destination.
\begin{figure}[h]
    \centering
    \tikzstyle{vertex}=[circle,fill=black!25,minimum size=20pt,inner sep=0pt]
\tikzstyle{smallvertex}=[circle,fill=black,minimum size=6pt,inner sep=0pt]
\tikzstyle{edge} = [draw,thin,->]
\tikzstyle{weight} = [font=\small]
\tikzstyle{selected edge} = [draw,line width=3pt,-,red!50]
\tikzstyle{big edge} = [draw,line width=2pt,->,black]
\begin{tikzpicture}[scale=0.55,auto,swap]
    \node[vertex](s1) at (0,0){$s$};
    \node[vertex](s2) at (5,0){$v_1$};
    \node[vertex](s3) at (10,0){$v_2$};
    \node[vertex](s4) at (15,0){$d$};

    \draw [edge] (s1) to[out=60,in=120, distance=2cm ] node[weight,below,black]{$e_1^1$} (s2);
    \draw [edge] (s1) to[out=20,in=160, distance=1.3cm ] node[weight,below,black]{$e_2^1$} (s2);
    \draw [edge] (s1) to[out=-20,in=200, distance=1.3cm ] node[weight,below,black]{$e_3^1$} (s2);
    \draw [edge] (s1) to[out=-60,in=240, distance=2cm ] node[weight,below,black]{$e_4^1$} (s2);
    
    \draw [edge] (s2) to[out=45,in=135, distance=1cm ] node[weight,below,black]{$e_1^2$} (s3);
    \draw [edge] (s2) to[out=-45,in=-135, distance=1cm ] node[weight,below,black]{$e_2^2$} (s3);

    \draw [edge] (s3) to[out=0,in=180, distance=1cm ] node[weight,below,black]{$e_1^3$} (s4);
 \end{tikzpicture}
    \caption{The transit time of edge $e_4^1$ is $4$. All other edges have a transit time of $1$. Network, in which there is a state for $9$ players, where player 9 is strictly earlier at node $v_1$ than player 8, while simultaneously all players arrive as early as possible at the destination $d$.}
    \label{fig:overtakingInNE}
\end{figure}
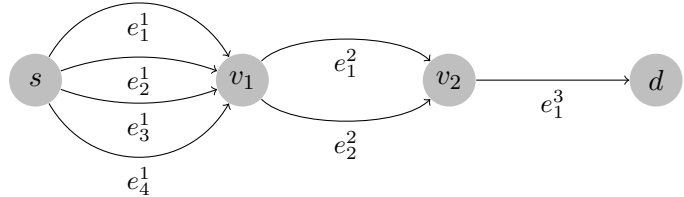
\subsection{Omitted Proof of Section \ref{sec:understandingopt}}\label{app:understandingopt}

\satzOptStructure*
\noindent
\begin{proof}
    Ford and Fulkerson~\cite{ford1958constructing} showed that a packet routing maximizing the number of packets arriving at $d$ until $T$ in linear multigraphs with edge capacity equal to one can be computed by deleting shortest paths as long as the distance $dist(s,d)$ is smaller than $T$.
        \begin{enumerate}
        \item Initialize $x_P=0$ for all $P\in \mathcal{P}$.
        \item WHILE $dist(s,d)\leq T$:\\
        $\qquad$ Compute a shortest $s-d$ path $P$ in $G$, delete $P$ from $G$ and set $x_P = 1$.
        \item Return a temporally repeated flow $\tilde{S}$ for $(x_P)_{P\in \mathcal{P}}.$
        \end{enumerate}
    The temporally repeated flow $\tilde{S}$ is defined as follows. For all paths $P$ with $x_P = 1$ it sends packets into $P$ at time steps $\{0, 1, \dots, T-\tau(P)\}$. It is worth noting that for arbitrary graphs one needs to base the calculations above on the residual graphs. Here, we could simplify the algorithm as linear multigraphs are series-parallel and a shortest $s-d$ path would never use any backward edge in the residual graph.
    
    It is immediate to see that when instead sending $T-\tau(P)$ packets into $P$ at time $0$ and stopping afterwards, all packets immediately enter the queue of the first link. We obtain a state of the packet routing game where no packet waits outside of the first layer. 
    
    To finish the proof, fix an arbitrary optimal state $S^*$ for the packet routing game with $n$ players with latest arrival time $C(S^*)$. By~\cite{ford1958constructing} there is a temporally repeated flow $\tilde{S}$ for time horizon $C(S^*)$ sending $n'\geq n$ packets to $d$ up to time $C(S^*)$. If $n'> n$, we have $n' = j + n$ for some $1\leq j\leq k-1$, since otherwise $\Tilde{S}^*$ would not be optimal. To send the correct amount of packets, we delete the last packet from each of the paths $P^k,P^{k-1},\ldots, P^{k-j+1}$. Since $\tilde{S}$ is feasible, the obtained flow is also feasible. We have constructed an optimal state that fulfills all claimed properties.
    %
\end{proof}

\subsection{A Model Extension to Arbitrary Edge Capacities}\label{sec:edgcap}
We can generalize the model by allowing edges $e$ to have capacities $\nu_e\in \N_{>0}$ which may differ from one. In contrast to the model presented in the Section \ref{sec:prelim}, in the network loading process at any point in time, $\nu_e$ players are now allowed to leave the queue $q_e$. In this section, we will argue, with the help of a reduction, that the $\PoA$ and the $\PoS$ are still bounded as in the base model.

Given a game $\Gamma$ with edge capacities that may differ from one, we can construct a game $\Gamma'$ by replacing every edge $e=(v,w)$ with transit time $\tau(e)$ and capacity $\nu_e$ by $\nu_e$ many edges $e^1, \ldots,e^{\nu_e}=(v,w)$ each with transit time $\tau(e)$ and unit capacity.
For an arbitrary equilibrium $S$ of $\Gamma$, we can construct a strategy profile $S'$ of $\Gamma'$ as follows: for each player $i$, we replace in the player's strategy every edge $e\in P_{(i)}$ by $e^j$, where $j=((q_e(i)-1 )\mod \nu_e)+1$. In this context, $q_e(i)$ denotes the queuing position of player $i$ on edge $e$ when she uses this edge. As a result of this construction, the arrival times at every node and, consequently, the completion times are identical, that is, $C^{\Gamma}(S)=C^{\Gamma'}(S')$. Furthermore, $S'$ is an equilibrium of $\Gamma'$. If this were not the case, there would exist a player who could strictly improve her completion time by changing her strategy. By changing her strategy in $S$ to the corresponding edges she achieves a better completion time in $\Gamma$. This contradicts the fact that $S$ is an equilibrium. Therefore, we have 
$$\sup_{S\in \Nash^{\Gamma}} C^{\Gamma}(S) \leq \sup_{S'\in \Nash^{\Gamma'}} C^{\Gamma'}(S').$$
Let $S^*$ be an optimal state of $\Gamma$ and $(S')^*$ be the strategy profile of $\Gamma'$ that arises from the transformation of $S^*$ as described above. The optimal completion time in $\Gamma'$ is at most as big as $C^{\Gamma'}((S')^*)=C^{\Gamma}(S^*)$. 
In total, we obtain a game $\Gamma'$ with unit capacities for every game $\Gamma$ with $\PoA(\Gamma)\leq \PoA(\Gamma')$. Hence, we can apply Theorem \ref{satz_PoA2} to conclude that
\begin{align*}
    \PoS(\mathcal{H}) \leq \PoA(\mathcal{H}) = \sup\limits_{\Gamma: G(\Gamma) \in \mathcal{H}} \PoA(\Gamma) \leq \sup\limits_{\Gamma': G(\Gamma') \in \mathcal{H}, \nu_e = 1 \forall e} \PoA(\Gamma') \leq 2.
\end{align*}
In conjunction with Theorem \ref{satz_PoAgeqee1}, this leads to the conclusion that for the extended model with edge capacities, it holds that $$\frac{e}{e-1}\leq PoS(\mathcal{G}) \leq \PoA(\mathcal{G})\leq 2.$$ 

\subsection{Omitted Proofs of Section \ref{sec:poa}}\label{app:sec4}
\lemmaAli*
\begin{proof}
As defined in the network loading, we refer to an edge as \emph{used (at time $t$)} when a player leaves the queue of this edge at time $t$.
If, at any time $t$, there is a used edge $e^1_j$ that has a strictly higher 
workload than another used edge $e^1_i$ with $i<j$, then the player with the largest index queuing on edge $e^1_j$ could have chosen edge $e^1_i$ without worsening her latency since the queue lengths on all used edges decrease at the same rate. This contradicts the assumption that the player behaves according to \badNash.

Conversely, suppose at some time $t$, there is a used edge $e^1_j$ and another used edge $e^1_i$ with $i<j$ and $l_{e^1_j}(\badNash,t)+1<l_{e^1_i}(\badNash,t)$.
Then the player with the largest index to queue on edge $e_i$ could have chosen edge $e_j$ instead. This contradicts the assumption that we are in an equilibrium.
Together this yields $l_{e^1_j}(\badNash,t) \leq l_{e^1_i}(\badNash,t) \leq l_{e^1_j}(\badNash,t)+1$.
\end{proof}

\lemmaOneLayerNew*
\begin{proof}
In the network loading at time $t$, we first add every player entering an edge to the queue and then remove the first player from each queue. For any $S \in \Nash$, we denote the sum of queue lengths in the network at time $t$ after the removal by $Q(S,t)$. When comparing $Q(S,t)$ to $Q(S,t-1)$ observe that we can obtain $Q(S,t)$ from $Q(S,t-1)$ as follows. We denote the number of players starting at time $t$ with $Q^+(S,t)$. For every player with starting time $t$, we increase the sum of queue lengths by one. Afterwards, the total queue length decreases by the number of used edges at time $t$, i.e., edges with a non-empty queue, denoted by $Q^-(S,t)$. Thus, we have $Q(S,t)= Q(S,t-1)+Q^+(S,t)-Q^-(S,t)$.
\paragraph{Claim:} $Q(\badNash,t) \geq Q(S,t)$ for all $S\in \Nash$ and for all $t\in \N_{>0}$.

\noindent \textit{Proof of claim.}
Suppose there exists a state $S\in \Nash$ and a time $t\in \N_{>0}$ such that $Q(\badNash,t) < Q(S,t)$. 
Consider the earliest point in time $t'$ at which this happens. Thus, we know that $Q(\badNash,t'-1) \geq Q(S,t'-1)$ and $Q^+(S,t')= Q^+(\badNash,t')$.
We conclude that $Q^-(\badNash,t') > Q^-(S,t')$, i.e., $\badNash$ uses strictly more edges at $t'$ than $S$ does.
We call an edge $e$ \emph{new} (at time $t'$) if the queue $q_e$ at time $t$ is empty, i.e., the first player entering the edge experiences no queuing time on this edge. Conversely, we call an edge \emph{old} (at time $t$) if it is not new at time $t$. Furthermore, the first player on a new edge does not contribute to $Q(S,t')$, since this player is counted in $Q^+(S,t')$ as well as in $Q^-(S,t')$. 

Hence, $\badNash$ uses more new edges than $S$. Let $e'$ be a new edge exclusively used by $\badNash$ at $t'$. By Lemma \ref{lemma_ali}, \badNash\ also uses all edges with transit time smaller than $\tau(e')$ at time $t'$. Furthermore, $S$ cannot use an edge with a larger transit time than $\tau(e')$ without contradicting the equilibrium property. Since $S$ has more players in queues, there has to be an edge $e''$ which has more players in its queue under $S$ than under $\badNash$ at time $t'$.

Due to Lemma \ref{lemma_ali}, we know that in $\badNash$ new edges are only used if each edge with a smaller index has a strictly higher 
workload. Additionally, in $S$, the queue on $e''$ is even longer than in \badNash. Summarizing, $l_{e''}(S,t')\ge l_{e'}(S,t')+2 = \tau(e')+2$.
Consequently, the last player on $e''$ in $S$, which experiences a latency of at least $\tau(e')+1$, can improve her latency by changing her strategy to $e'$. This contradicts the assumption that $S$ is an equilibrium.\qed

\medskip
\noindent Furthermore, we observe that the completion time of any player $i$ starting at $t+1$ is uniquely defined given the total queue length $Q(S,t)$ at time $t$ for any $S\in \Nash$, since the $Q(S,t)$-many players in queues level out the latencies on the shortest edges. This suffices to determine the latencies of the subsequent players. Additionally, the completion time of player $i$ is monotonically increasing in $Q(S,t)$. This yields $C_i(\badNash)\geq C_i(S)$ for all $S \in \Nash$ and all $i\in N$.
\end{proof}

\lemmaMnKStartpattern*
\noindent

\lemmaAddingEdges*
\noindent
\begin{proof}
It is sufficient to show that the property holds for $|E_2\backslash E_1|=1$. Let $e=(v_j,v_{j+1})$ be the edge where the two graphs differ. Up to node $v_j$ the arrival patterns $a_{v_j}(\badNash_{G_1})$ and $a_{v_j}(\badNash_{G_2})$ trivially coincide.
Suppose that there is a player $i$ such that $a^i_{v_{j+1}}(\badNash_{G_2}) > a^i_{v_{j+1}}(\badNash_{G_1})$ and $a^{i'}_{v_{j+1}}(\badNash_{G_2}) \leq a^{i'}_{v_{j+1}}(\badNash_{G_1})$ for all $i'<i$. In particular, there are at most $|\{e\in E_1^{j+1}: \tau(e) \leq a^i_{v_{j+1}}(\badNash_{G_1}) - a^i_{v_{j}}(\badNash_{G_1})\}| -1$ players with index less than $i$ that arrive at $v_{j+1}$ at time $a^i_{v_{j+1}}(\badNash_{G_1})$. Here $E_1^{j+1}$ denotes the edges in $E_1$ that are in the $(j+1)$-th layer of $G_1$. Since $i$ arrives at $v_{j+1}$ strictly later in $G_2$ than in $G_1$, among all edges of $\{e\in E_1^{j+1}: \tau(e) \leq a^i_{v_{j+1}}(\badNash_{G_1}) - a^i_{v_{j}}(\badNash_{G_1})\}$ players arrive at $v_{j+1}$ at time $a^i_{v_{j+1}}(\badNash_{G_1})$ in $\badNash_{G_2}$. Note, that is one more player. By Lemma \ref{lemma_order} the additional player $i'$ has $i'<i$, since no overtaking is possible at any intermediate node in any equilibrium. But then we have $a^{i'}_{v_{j+1}}(\badNash_{G_2}) > a^{i'}_{v_{j+1}}(\badNash_{G_1})$, which contradicts that player $i$ was the player with the smallest index that was delayed. Therefore, $a^i_{v_{j+1}}(\badNash_{G_1}) \geq a^i_{v_{j+1}}(\badNash_{G_2})$ holds for all $i\in N$.
In the later layers, the graphs are indistinguishable. Therefore, by iterating node by node to the destination node $d$, we obtain the statement by applying Lemma $\ref{lemma_mnK_startpattern}$ at each step.
\end{proof}


\subsection{Basic Notations for Flows over Time and Omitted Proof of Section \ref{sec:FlowOverTime}}\label{app:FlowOverTime}

Flows over time can essentially be seen as the continuous variant of packet routings. For a comprehensive introduction to the topic, we refer to the survey article of Skutella~\cite{Skutella2009survey}. We will follow its lines and recall the most basic definitions here. For a given graph $G$ with transit times $\tau(e)$ for the edges $e \in E$, a designated source $s$, destination $d$ and a fixed time horizon $T$, a flow over time is formally defined as follows.

\begin{definition}
A flow over time $f$ with time horizon $T$ consists of a Lebesgue-
integrable function $f_e : [0, T) \rightarrow \mathbb{R}_{\geq 0}$ for each arc $e \in E$. For all $\theta \geq T - \tau(e)$ it must hold that $fe(\theta) = 0$ for all $e \in E$.
\end{definition}

\begin{definition}
    A flow over time $f$ is called feasible if the following properties hold.
    \begin{enumerate}
    \item The flow over time $f$ fulfills the capacity constraints if $f_e(\theta) \leq \nu_e$ for
each $e \in E$ and almost all $\theta \in [0, T )$.
    \item The flow over time $f$ fulfills the weak flow conservation constraints if the amount of flow that has left a node $v \in V\setminus \{s,d\}$ is always upper bounded by the amount of flow that has entered the same node for all $\theta \in [0,T)$. Formally,
    \[\sum_{e \in \delta^-(v)}\int_0^{\theta - \tau(e)}f_e(\xi)\,d\xi \geq \sum_{e \in \delta^+(v)}\int_0^{\theta}f_e(\xi)\,d\xi\;,\]
    where $\delta^-(v)$ and $\delta^+(v)$ denote the set incoming and outgoing of edges of $v$, respectively.
    \end{enumerate}
\end{definition}
\noindent We say that the amount of flow that has entered $d$ up to time $T$ is the value of the flow $f$. Formally, $|f| = \sum_{e \in \delta^-(d)}\int_0^{T - \tau(e)}f_e(\xi)\,d\xi$. If the weak flow conservation constraint is fulfilled with equality for all $\theta$ and all $v \in V\setminus\{s,d\}$, we say the flow fulfills strong flow conservation. A very special and important class of flows over time are temporally repeated flows. For a feasible static flow $x$ on the same graph, let $(x_P)_{P\in \mathcal{P}}$ be a path decomposition of $x$. The temporally repeated flow $f$ sends flow at rate $x_P$ into $P$ from $s$ during the time interval $[0,T-\tau(P))$. Here, $\tau(P)\coloneqq \sum_{e \in P}{\tau(e)}$ denotes the length of the path $P$. Formally,

\begin{definition}
For a static flow $x$ with flow decomposition $(x_P)_{P \in \mathcal{P}}$ the corresponding temporally repeated flow $f$ with time horizon $T$ is defined by
\[f_e(\theta) \coloneqq \sum_{P \in P_e(\theta)} x_p \qquad \text{for }e=(v,w) \in E, \theta \in [0,T),\]
where $$P_e(\theta) \coloneqq \left\{P \in \mathcal{P} : e \in P \wedge \tau(P_{s,v}) \leq \theta \wedge \tau(P_{v,t}) < T-\theta \right\}.$$ Here, $\tau(P_{u,v})$ denotes the length of $P$ from $u$ to $v$.
\end{definition}
\noindent It is easy to see that temporally repeated flows are feasible and fulfill strong flow conservation, see~\cite{Skutella2009survey}.

A dynamic Nash equilibrium, also called a Nash flow over time, is a special flow over time. It is assumed, that flow appears with a fixed inflow rate at the source node. Similar to the discrete model, the flow is interpreted as an infinite amount of flow particles each deciding at the time of appearance at the source node on a shortest $s-d$ path in a dynamic Nash equilibrium. More formally, for any point in time $\theta \in \mathbb{R}_{\geq 0}$ a positive flow entering an edge $e \in E$ implies that $e$ lies on a currently fastest path to $d$. The makespan-$\PoA$ and makespan-$\PoS$ are defined analogously to the packet routing game, i.e., we compare the latest arrival time of any particle in a Nash flow over time to the arrival time in an optimal flow. Note that for Nash flows over time we even have $\PoA(\mathcal{G})=\PoS(\mathcal{G})$ since the arrival times of dynamic Nash flows over time are unique as recently shown by Olver et al.~\cite{DBLP:conf/focs/OlverSK21}.
For a formal introduction to Nash flows over time, we refer to the PhD theses of Laura Vargas Koch~\cite{PhDLaura} and Leon Sering~\cite{PhDLeon}.

\propositionFlowOverTimeLowerBound*
\begin{proof}
We adapt the graphs used in the proof of Theorem \ref{satz_PoAgeqee1} by omitting the first layer and introducing a constant inflow of $k$ at $s$. We interpret the $n$ players as a flow mass of $n$ similar to the interpretation used by Fleischer and Tardos~\cite{DBLP:journals/orl/FleischerT98}. It remains to argue that the obtained flow over time is indeed feasible and a Nash flow.

Starting with an arbitrary equilibrium $S_i\in \Nash$ of $\Gamma_i$, we construct the Nash flow over time as follows. We set $f_e(\theta)=1$ in the interval \mbox{$[t-1,t)$} whenever a packet leaves the queue $q_e$ of an edge $e=(u,v)$ at time $t$. The constructed flow over time is feasible as we only forward packets if they have already arrived at the corresponding node, thus we immediately respect weak flow conservation in the flow over time. Additionally, as there is at most one packet leaving any queue at a given time, thus we respect the capacity constraints. Regarding the Nash condition observe that at each discrete point in time, a standard edge of some layer is used if and only if all other standard edges are also used as the total number of packets is divisible by any \mbox{$j\in \{1,\ldots, k\}$}. Thus, the standard edges of a layer are used equally, and it is immediate that no standard edge has a strictly lower workload than some other standard edge of the same layer. Additionally, the special edges have strictly larger workloads and are never part of the shortest path network. Since the first omitted layer consists exclusively of standard edges with a transit time of one, we respect the inflow into the network. The completion time of the dynamic Nash equilibrium is by definition $C(S_i)$.

It remains to bound the completion time of an optimal flow over time in this instances. In series-parallel graphs, there is a temporally repeated flow over time that is also an earliest arrival flow, i.e., a flow that maximizes the flow that arrived at $d$ at any point in time, see~\cite{Skutella2009survey}. In particular, this flow is optimal for minimizing the completion time. By Proposition~\ref{thm:opt} the completion time of this earliest arrival flow and the optimal packet routing coincide up to one and we conclude in the limit $\PoA(\mathcal{G})\geq \frac{e}{e-1}$.
\end{proof}

\fi
\end{document}